\newif\iffigsone
\newif\iffigstwo
\newif\iffigsthree
\figsonetrue
\figstwotrue
\figsthreetrue

\documentclass{article}
\usepackage{amsmath,amssymb,amstext,amsthm}
\usepackage{nicematrix}
\usepackage{tikz}
\usepackage{minted}
\usepackage[pdftex,pagebackref=false]{hyperref}
\usepackage{keytheorems}
\usepackage[width=.9\textwidth]{caption}
\usepackage{booktabs}
\usepackage{authblk}
\usepackage{braket}
\usepackage{nicematrix}
\usepackage[outline]{contour}
\usepackage[style=numeric,maxnames=50,backend=biber]{biblatex}
\usepackage[Tol]{colorblind}

\DeclareSourcemap{
  \maps[datatype=bibtex]{
    \map{
      \step[fieldsource=doi,final]
      \step[fieldset=url,null]
      \step[fieldset=urldate,null]
      \step[fieldset=eprint,null]
    }
    \map{
      \step[fieldsource=archivePrefix,final]
      \step[fieldset=url,null]
      \step[fieldset=urldate,null]
    }
  }
}

\hypersetup{
    plainpages=false,
    unicode=false,
    pdftoolbar=true,
    pdfmenubar=true,
    pdffitwindow=false,
    pdfstartview={FitH},
    pdfnewwindow=true,
}

\newcommand{\zerodel}{.\kern-\nulldelimiterspace}
\newcommand{\sta}[1]{\color{white}\contour{black}{$#1$}}
\newcommand{\stb}[1]{\color{black}\contour{black}{$#1$}}

\newkeytheorem{theorem}[parent=section]
\newkeytheorem{lemma}[sibling=theorem]
\newkeytheorem{definition}[sibling=theorem,style=definition]

\title{Beyond the Magic Square Game:\\Widening the Gap for Two Bell States}
\author{Tony Lau}
\date{\today}

\begin{document}

\maketitle

\begin{abstract}
Prior to this work, the largest known gap between the entangled value and the classical value for a one-round two-player nonlocal game with a perfect entangled strategy using two Bell states of entanglement was $\frac{1}{9}$, achieved by the Mermin-Peres magic square game. A larger gap of $\frac{2}{15}$ has been claimed for the related doily game of Kelleher, Roomy, and Holweck (Quant. Inf. Proc., 2024). However, this value assumes that the classical players use identical local strategies, and we show that dropping this assumption reduces the gap back to $\frac{1}{9}$. We then demonstrate that $\frac{1}{9}$ is not optimal by explicitly constructing a new nonlocal game with gap~$\frac{4}{35}$.
\end{abstract}

\section{Introduction}

The study of nonlocal games has been instrumental in the advancement of quantum information theory, having broad applications from Bell inequalities \cite{PhysicsPhysiqueFizika.1.195}\cite{PhysRev.47.777}, to quantum key distribution in cryptography \cite{PhysRevLett.98.230501}\cite{PhysRevLett.97.120405}, and to multi-prover interactive proofs in complexity theory \cite{1313847}\cite{10.1145/3485628}.

This work considers nonlocal games in the context of device-independent self-testing (see \cite{Supic2020selftestingof} and the references therein). In general, it is known that the size of Bell violations may be increased by using larger question sets \cite{Cirelson1980} or using larger answer sets \cite{Junge2010}, but these results do not constrain the amount of entanglement shared between the entangled players. It remains unknown to what extent increasing the size of the question and answer sets may improve the distinguishability between fixed amounts of entanglement. For simplicity, we require that the entangled players must win with certainty.

The Mermin-Peres magic square game \cite{PhysRevLett.65.3373}\cite{PERES1990107} is a well-studied game that distinguishes two Bell states from no entanglement. This game may be won with probability 1 only if the players do share two Bell states, while the game may be won with probability at most $\frac{8}{9}$ if they share no entanglement. Kelleher, Roomy, and Holweck introduce the \emph{doily game} \cite{Kelleher2024} that extends the symmetries of the magic square game to the full 2-qubit Pauli group, and they compute a classical value of $\frac{13}{15}$ for this game. However, their analysis contains an implicit assumption that the players use identical local strategies, which is not generally true for all classical strategies. We provide an explicit classical strategy in the proof of \Autoref{thm:ams_val} that achieves a success probability of $\frac{8}{9} > \frac{13}{15}$ for the doily game, which is the same as the classical value of magic square game. This opens the question: is the magic square game indeed optimal for distinguishing two Bell states from no entanglement?

In fact, no. In this work, we widen the entangled-classical gap for two Bell states from $\frac{1}{9}$ to $\frac{4}{35}$ by explicitly constructing a variant of the doily game augmented with synchronous questions that has a classical value of $\frac{31}{35}$ and a two-Bell-state entangled value of 1. \autoref{fig:val_table} summarizes this paper's results.

\begin{table}
    \centering
    \renewcommand{\arraystretch}{1.3}
    \begin{tabular}{c|c|c}
         Game & Classical value & Entangled value \\ \midrule
         MS game (\cite{PhysRevLett.65.3373}\cite{PERES1990107}) & $\frac{8}{9}$ & 1 \\
         Doily game (\Autoref{thm:ams_val}) & $\frac{8}{9}$ & 1 \\
         $\frac{1}{7}$-sync doily game (\Autoref{thm:psams}) & $\frac{31}{35}$ & 1 \\
    \end{tabular}
    \caption{Games considered in this work, their classical values, and their entangled values. Note that $\frac{31}{35} < \frac{8}{9}$.}
    \label{fig:val_table}
\end{table}

\section{Preliminaries} \label{sec:prel}

We define a one-round two-player nonlocal game as a cooperative game between two players, commonly referred to as Alice and Bob, who are each asked a single question and each provide a single response, and then those responses are jointly scored by a referee.
\begin{definition}
A \emph{one-round two-player nonlocal game} $G$ is defined as the tuple $G = (\mathcal{X}, \mathcal{Y}, \mathcal{A}, \mathcal{B}, \pi, V)$, in which $\mathcal{X}$ and $\mathcal{Y}$ are the sets of possible questions to Alice and Bob, respectively, $\mathcal{A}$ and $\mathcal{B}$ are the sets of possible answers from Alice and Bob respectively, $\pi: \mathcal{X} \times \mathcal{Y} \to [0, 1]$ is a probability distribution over the question sets, and $V: \mathcal{X} \times \mathcal{Y} \times \mathcal{A} \times \mathcal{B} \to \{0, 1\}$ is the referee's payoff function.
\end{definition}

We then define a classical deterministic strategy for a nonlocal game as any pair of functions, the first mapping Alice's questions to Alice's answers, and the second mapping Bob's questions to Bob's answers.
\begin{definition}
For the nonlocal game $G = (\mathcal{X}, \mathcal{Y}, \mathcal{A}, \mathcal{B}, \pi, V)$, a \emph{classical deterministic strategy} $S$ is any pair of functions $S = (A, B)$ with $A: \mathcal{X} \to \mathcal{A}$ and $B: \mathcal{Y} \to \mathcal{B}$. $A$ and $B$ will be referred to as Alice's local strategy and Bob's local strategy, respectively. Then, the probability that Alice and Bob win $G$ using $S$ is given by the expression
\begin{align}
    \sum_{\substack{x \in \mathcal{X} \\ y \in \mathcal{Y}}} \pi(x, y) V(x, y, A(x), B(y)).
\end{align}
\end{definition}

Similarly, we define an entangled strategy for a nonlocal game with local dimension $d$ as a triplet of a shared entangled state with Schmidt rank $d$, a function mapping Alice's questions to POVM measurements with Alice's answers as outcomes, and a function mapping Bob's questions to POVM measurements with Bob's answers as outcomes.
\begin{definition}
For the nonlocal game $G = (\mathcal{X}, \mathcal{Y}, \mathcal{A}, \mathcal{B}, \pi, V)$, an \emph{entangled strategy} $S$ with local dimension $d$ is any triplet $S = (\ket{\psi}, A, B)$ with $\ket{\psi} \in \mathbb{C}^{2d}$, $\{(A_{x,a}, a) : a \in \mathcal{A}\}$ being a POVM for every $x \in \mathcal{X}$, and $\{(B_{y,b}, b) : b \in \mathcal{B}\}$ being a POVM for every $y \in \mathcal{Y}$. As with classical strategies, $A$ and $B$ will be referred to as Alice's local strategy and Bob's local strategy, respectively. Then, the probability that Alice and Bob win $G$ using $S$ is given by the expression
\begin{align}
    \sum_{\substack{x \in \mathcal{X} \\ y \in \mathcal{Y}}} \sum_{\substack{a \in \mathcal{A} \\ b \in \mathcal{B}}} \pi(x, y) \Braket{\psi | A_{x,a} \otimes B_{y,b} | \psi} V(x, y, a, b).
\end{align}
\end{definition}

We may then define the \emph{classical value} and \emph{entangled value} of a game $G$ as the supremal probability of winning $G$ over all classical strategies and over all entangled strategies, respectively.

\subsection{The magic square (MS) game} \label{sec:ms}

The \emph{magic square game} \cite{PhysRevLett.65.3373}\cite{PERES1990107} is a particularly well-studied nonlocal game. The precise definition of the magic square game varies in the literature, but the core idea of the game is that Alice and Bob must output entries of a three-by-three boolean magic square such that the rows and columns obey certain parity constraints, and Alice's answer and Bob's answer must agree on the entries that are shared between their outputs (see \cite{10.1119/1.1773173}).

In this work, we consider a specific form of the magic square game. For questions, Alice and Bob each receive a row or column, chosen uniformly randomly such that Alice's row or column intersects with Bob's row or column on exactly one entry. Alice and Bob must then each produce three bits subject to two constraints:
\begin{enumerate}
    \item \emph{Parity}: if the question corresponds to a row, the parity of the three bits must be even, whereas, if the question corresponds to a column, the parity of the three bits must be odd.
    \item \emph{Consistency}: the bits corresponding to the entry that is shared between the two players must be the same.
\end{enumerate}

We define the magic square game formally as follows:
\begin{definition}
\sloppy The \emph{magic square (MS) game} is defined as the nonlocal game $([6], [6], \{0, 1\}^3, \{0, 1\}^3, \pi, R)$, in which $\pi$ is the uniform distribution over
\begin{align}
    \{(j, k) \in [6] \times [6] : ~\text{$E_j$ and $E_k$ have exactly one variable in common}\}
\end{align}
and $R: [6] \times [6] \times \{0, 1\}^3 \times \{0, 1\}^3 \to \{0, 1\}$ is the function such that $R(j, k, a, b) = 1$ if and only if $a \in \{0, 1\}^3$ satisfies the parity of $E_j$, $b \in \{0, 1\}^3$ satisfies the parity of $E_k$, and the bit corresponding to the variable shared between $E_j$ and $E_k$ in $a$ is the same as that in $b$.

For the MS game, we write $A_m(j)$ to denote Alice's bit output for $V_m$ when asked equation $E_j$, and the same for $B_m(j)$ and Bob.
\end{definition}

Naturally, the magic square cannot be satisfied with classical bit assignments, as the row parity constraints require that the parity of the entire magic square is even, while the column parity constraints require that the parity of the entire magic square is odd. At most eight out of the nine entries can be satisfied consistently, so the optimal classical strategy for the MS game wins with probability $\frac{8}{9}$.

Remarkably, entangled strategies may win the MS game with probability~1, using two Bell states. If we consider assigning operators to the entries of the magic square instead of bits, there exists an assignment of observables to the entries of the magic square such that each row and column consists of compatible observables and that measurement outcomes along each row or column satisfy that row's or that column's parity constraint. Specifically, we use triples of distinct, non-identity, mutually commuting observables from the two-qubit Pauli group $\tilde{P}_2 = \langle IX, XI, IZ, ZI \rangle$, which we refer to as \emph{commuting triples}. One set of such commuting triples is given with their associated bit-valued equations in \autoref{fig:ms_prods} and is represented diagrammatically in \autoref{fig:ms_diagr}.

\begin{table}[h]
    \centering
    {\setlength{\tabcolsep}{3pt}
    \begin{tabular}{@{\hspace{8pt}} l @{\hspace{8pt}} | @{\hspace{8pt}} lllllcc @{\hspace{9pt}} | @{\hspace{9pt}} cccccc @{\hspace{9pt}} r @{\hspace{8pt}}}
        Label & \multicolumn{7}{c @{\hspace{9pt}} | @{\hspace{9pt}}}{Boolean equation} & \multicolumn{7}{c @{\hspace{8pt}}}{Commuting triple} \\ \midrule
        \hspace{2pt} $E_0$ & $V_0$ & + & $V_1$ & + & $V_2$ & = & 0 & $XX$ & $\cdot$ & $YZ$ & $\cdot$ & $ZY$ & = & $II$ \\
        \hspace{2pt} $E_1$ & $V_3$ & + & $V_4$ & + & $V_5$ & = & 0 & $YY$ & $\cdot$ & $ZX$ & $\cdot$ & $XZ$ & = & $II$ \\
        \hspace{2pt} $E_2$ & $V_6$ & + & $V_7$ & + & $V_8$ & = & 0 & $ZZ$ & $\cdot$ & $XY$ & $\cdot$ & $YX$ & = & $II$ \\
        \hspace{2pt} $E_3$ & $V_0$ & + & $V_3$ & + & $V_6$ & = & 1 & $XX$ & $\cdot$ & $YY$ & $\cdot$ & $ZZ$ & = & $-II$ \\
        \hspace{2pt} $E_4$ & $V_1$ & + & $V_4$ & + & $V_7$ & = & 1 & $YZ$ & $\cdot$ & $ZX$ & $\cdot$ & $XY$ & = & $-II$ \\
        \hspace{2pt} $E_5$ & $V_2$ & + & $V_5$  & + & $V_8$ & = & 1 & $ZY$ & $\cdot$ & $XZ$ & $\cdot$ & $YX$ & = & $-II$
    \end{tabular}}
    \caption{The magic square game's six boolean equations and their associated commuting triples. If the variable are placed on a three by three grid, then $E_0$, $E_1$, and $E_2$ would be rows and $E_3$, $E_4$, and $E_5$ would be columns.}
    \label{fig:ms_prods}
\end{table}
\begin{figure}
    \centering
    \begin{tikzpicture}[scale=0.75]
    \coordinate (V0) at (-3,3);
    \coordinate (V1) at (0,3);
    \coordinate (V2) at (3,3);
    \coordinate (V3) at (-3,0);
    \coordinate (V4) at (0,0);
    \coordinate (V5) at (3,0);
    \coordinate (V6) at (-3,-3);
    \coordinate (V7) at (0,-3);
    \coordinate (V8) at (3,-3);
    
    \draw[ultra thick] (V0) -- (V2);
    \draw[ultra thick] (V3) -- (V5);
    \draw[ultra thick] (V6) -- (V8);
    
    \draw[red,ultra thick] (V0) -- (V6);
    \draw[red,ultra thick] (V1) -- (V7);
    \draw[red,ultra thick] (V2) -- (V8);
    
    \draw[ultra thick,fill=white] (V0) circle [radius=0.75];
    \draw[ultra thick,fill=white] (V1) circle [radius=0.75];
    \draw[ultra thick,fill=white] (V2) circle [radius=0.75];
    \draw[ultra thick,fill=white] (V3) circle [radius=0.75];
    \draw[ultra thick,fill=white] (V4) circle [radius=0.75];
    \draw[ultra thick,fill=white] (V5) circle [radius=0.75];
    \draw[ultra thick,fill=white] (V6) circle [radius=0.75];
    \draw[ultra thick,fill=white] (V7) circle [radius=0.75];
    \draw[ultra thick,fill=white] (V8) circle [radius=0.75];
    
    \node at (V0) {\large $XX$};
    \node at (V1) {\large $YZ$};
    \node at (V2) {\large $ZY$};
    \node at (V3) {\large $YY$};
    \node at (V4) {\large $ZX$};
    \node at (V5) {\large $XZ$};
    \node at (V6) {\large $ZZ$};
    \node at (V7) {\large $XY$};
    \node at (V8) {\large $YX$};
    
    \node at (-4.5,3) {\Large $E_0$};
    \node at (-4.5,0) {\Large $E_1$};
    \node at (-4.5,-3) {\Large $E_2$};
    
    \node at (-3,4.5) {\Large $E_3$};
    \node at (0,4.5) {\Large $E_4$};
    \node at (3,4.5) {\Large $E_5$};
    \end{tikzpicture}
    \caption{The nine variables (circles) and six equations (lines) of the magic square game embedded in a planar graph along with their operator solution. The black lines have parity 0 while the red lines have parity 1.}
    \label{fig:ms_diagr}
\end{figure}

Note that we can construct a family of games such that each game is logically equivalent to the MS game by varying the parity constraints of the equations in a certain way. For example, consider flipping the parity constraints for the first row and first column ($E_0$ and $E_3$). An operator solution for this modified game may be obtained from an operator solution to the original game by multiplying the operator in the first column and first row (labeled by $V_0$) by $-II$. Together, the rows now require that the total parity of the magic square is odd while the columns now require that the total parity of the magic square is even, and the optimal classical winning probability is still $\frac{8}{9}$. Indeed, as long as the total parity of the magic square as required by the rows differs from the total parity as required by the columns, the game is equivalent to the MS game \cite{arkhipov2012}. We define the family of these games, including the MS game itself, as the \emph{MS-equivalent games}.

Moreover, the existence of one operator solution to an MS-equivalent game implies that the same operators, but each with an arbitrary choice of sign, form an operator solution to another MS-equivalent game. As such, for the rest of this work, we only consider commuting triples formed from the signless elements of the Pauli quotient group, $P_2 = \tilde{P}_2 / \langle iII \rangle$. However, since $P_2$ is abelian, we will continue to consider the products of operators as if they are from $\tilde{P}_2$ in regards to commutation relations and the combined product of each triple. In doing so, each commuting triple from $P_2$ uniquely identifies a maximal abelian subgroup of $\tilde{P}_2$.

Furthermore, we refer to any set of nine distinct elements of $P_2$ that form an operator solution to an MS-equivalent game as an \emph{operator magic square}. By the symmetries of the Pauli group, it is known that every set of nine distinct operators from $P_2$ that can be organized into three rows of commuting triples intersecting with three columns of commuting triples forms an operator magic square.

\section{The doily game} \label{sec:ams}

Noticing that $P_2$ has a richer structure than the MS game utilizes, Kelleher, Roomy, and Holweck introduced the \emph{doily game} \cite{Kelleher2024} to utilize the full symmetry of $P_2$. Specifically, the operator solution to the MS game only uses six commuting triples over nine elements of $P_2$, while there are in fact 15 distinct commuting triples over all 15 non-identity elements of $P_2$.

\begin{table}[h]
    \centering
    {\setlength{\tabcolsep}{3pt}
    \begin{tabular}{@{\hspace{8pt}} l @{\hspace{8pt}} | @{\hspace{8pt}} lllllcc @{\hspace{9pt}} | @{\hspace{9pt}} cccccc @{\hspace{9pt}} r @{\hspace{8pt}}}
        Label & \multicolumn{7}{c @{\hspace{9pt}} | @{\hspace{9pt}}}{Boolean equation} & \multicolumn{7}{c @{\hspace{8pt}}}{Commuting triple} \\ \midrule
        \hspace{2pt} $E_0$ & $V_0$ & + & $V_3$ & + & $V_6$ & = & 0 & $IX$ & $\cdot$ & $XI$ & $\cdot$ & $XX$ & = & $II$ \\
        \hspace{2pt} $E_1$ & $V_1$ & + & $V_3$ & + & $V_7$ & = & 0 & $IY$ & $\cdot$ & $XI$ & $\cdot$ & $XY$ & = & $II$ \\
        \hspace{2pt} $E_2$ & $V_2$ & + & $V_3$ & + & $V_8$ & = & 0 & $IZ$ & $\cdot$ & $XI$ & $\cdot$ & $XZ$ & = & $II$ \\
        \hspace{2pt} $E_3$ & $V_0$ & + & $V_4$ & + & $V_9$ & = & 0 & $IX$ & $\cdot$ & $YI$ & $\cdot$ & $YX$ & = & $II$ \\
        \hspace{2pt} $E_4$ & $V_1$ & + & $V_4$ & + & $V_{10}$ & = & 0 & $IY$ & $\cdot$ & $YI$ & $\cdot$ & $YY$ & = & $II$ \\
        \hspace{2pt} $E_5$ & $V_2$ & + & $V_4$ & + & $V_{11}$ & = & 0 & $IZ$ & $\cdot$ & $YI$ & $\cdot$ & $YZ$ & = & $II$ \\
        \hspace{2pt} $E_6$ & $V_0$ & + & $V_5$ & + & $V_{12}$ & = & 0 & $IX$ & $\cdot$ & $ZI$ & $\cdot$ & $ZX$ & = & $II$ \\
        \hspace{2pt} $E_7$ & $V_1$ & + & $V_5$ & + & $V_{13}$ & = & 0 & $IY$ & $\cdot$ & $ZI$ & $\cdot$ & $ZY$ & = & $II$ \\
        \hspace{2pt} $E_8$ & $V_2$ & + & $V_5$ & + & $V_{14}$ & = & 0 & $IZ$ & $\cdot$ & $ZI$ & $\cdot$ & $ZZ$ & = & $II$ \\
        \hspace{2pt} $E_9$ & $V_6$ & + & $V_{11}$ & + & $V_{13}$ & = & 0 & $XX$ & $\cdot$ & $YZ$ & $\cdot$ & $ZY$ & = & $II$ \\
        \hspace{2pt} $E_{10}$ & $V_{10}$ & + & $V_{12}$ & + & $V_8$ & = & 0 & $YY$ & $\cdot$ & $ZX$ & $\cdot$ & $XZ$ & = & $II$ \\
        \hspace{2pt} $E_{11}$ & $V_{14}$ & + & $V_7$ & + & $V_9$ & = & 0 & $ZZ$ & $\cdot$ & $XY$ & $\cdot$ & $YX$ & = & $II$ \\
        \hspace{2pt} $E_{12}$ & $V_6$ & + & $V_{10}$ & + & $V_{14}$ & = & 1 & $XX$ & $\cdot$ & $YY$ & $\cdot$ & $ZZ$ & = & $-II$ \\
        \hspace{2pt} $E_{13}$ & $V_{11}$ & + & $V_{12}$ & + & $V_7$ & = & 1 & $YZ$ & $\cdot$ & $ZX$ & $\cdot$ & $XY$ & = & $-II$ \\
        \hspace{2pt} $E_{14}$ & $V_{13}$ & + & $V_8$  & + & $V_9$ & = & 1 & $ZY$ & $\cdot$ & $XZ$ & $\cdot$ & $YX$ & = & $-II$
    \end{tabular}}
    \caption{The doily game's 15 boolean equations and their associated commuting triples. Note that the last six equations share an identical structure with \autoref{fig:ms_prods}.}
    \label{fig:ams_prods}
\end{table}
\begin{figure}
    \centering
    \begin{tikzpicture}[scale=0.75]
    \coordinate (V1) at (90:3);
    \coordinate (V3) at (162:3);
    \coordinate (V6) at (234:3);
    \coordinate (V14) at (306:3);
    \coordinate (V5) at (18:3);
    
    \coordinate (V10) at (270:2);
    \coordinate (V2) at (342:2);
    \coordinate (V13) at (54:2);
    \coordinate (V7) at (126:2);
    \coordinate (V0) at (198:2);
    
    \coordinate (V4) at (270:5);
    \coordinate (V8) at (342:5);
    \coordinate (V11) at (54:5);
    \coordinate (V9) at (126:5);
    \coordinate (V12)  at (198:5);
    
    \draw[ultra thick] (V1) -- (V3);
    \draw[ultra thick] (V3) -- (V6);
    \draw[ultra thick] (V14) -- (V5);
    \draw[ultra thick] (V5) -- (V1);
    
    \draw[ultra thick] (V1) -- (V4);
    \draw[ultra thick] (V3) -- (V8);
    \draw[ultra thick] (V6) -- (V11);
    \draw[ultra thick] (V14) -- (V9);
    \draw[ultra thick] (V5) -- (V12);
    
    \draw[ultra thick] (V4) -- (V11);
    \draw[ultra thick] (V12) -- (V8);
    \draw[ultra thick] (V9) -- (V4);
    
    \draw[red,ultra thick] (V6) -- (V14);
    \draw[red,ultra thick] (V11) -- (V12);
    \draw[red,ultra thick] (V8) -- (V9);
    
    \draw[ultra thick,fill=white] (V1) circle [radius=0.75];
    \draw[ultra thick,fill=white] (V3) circle [radius=0.75];
    \draw[ultra thick,fill=white] (V6) circle [radius=0.75];
    \draw[ultra thick,fill=white] (V14) circle [radius=0.75];
    \draw[ultra thick,fill=white] (V5) circle [radius=0.75];
    \draw[ultra thick,fill=white] (V10) circle [radius=0.75];
    \draw[ultra thick,fill=white] (V2) circle [radius=0.75];
    \draw[ultra thick,fill=white] (V13) circle [radius=0.75];
    \draw[ultra thick,fill=white] (V7) circle [radius=0.75];
    \draw[ultra thick,fill=white] (V0) circle [radius=0.75];
    \draw[ultra thick,fill=white] (V4) circle [radius=0.75];
    \draw[ultra thick,fill=white] (V8) circle [radius=0.75];
    \draw[ultra thick,fill=white] (V11) circle [radius=0.75];
    \draw[ultra thick,fill=white] (V9) circle [radius=0.75];
    \draw[ultra thick,fill=white] (V12) circle [radius=0.75];
    
    \node at (V1) {\large $IY$};
    \node at (V3) {\large $XI$};
    \node at (V6) {\large $XX$};
    \node at (V14) {\large $ZZ$};
    \node at (V5) {\large $ZI$};
    \node at (V10) {\large $YY$};
    \node at (V2) {\large $IZ$};
    \node at (V13) {\large $ZY$};
    \node at (V7) {\large $XY$};
    \node at (V0) {\large $IX$};
    \node at (V4) {\large $YI$};
    \node at (V8) {\large $XZ$};
    \node at (V11) {\large $YZ$};
    \node at (V9) {\large $YX$};
    \node at (V12) {\large $ZX$};
    \end{tikzpicture}
    \caption{The 15 variables (circles) and 15 equations (lines) of the doily game embedded in a planar graph along with their operator solution. Each circle has exactly three lines intersecting it, though not all intersections pass through a circle's center. The black lines have parity 0 while the red lines have parity 1.}
    \label{fig:ams_diagr}
\end{figure}

To construct the doily game, we may transform the commuting triples of $P_2$ and their products into a boolean linear system of 15 equations and 15 variables, as shown in \autoref{fig:ams_prods} and \autoref{fig:ams_diagr}. In particular, we associate each of the 15 non-identity elements of the Pauli quotient group with a variable, $V_0$ to $V_{14}$, and each commuting triple with an equation, $E_0$ to $E_{14}$, such that the sum of the variables in the equation is given by the product of the associated commuting triple.

As with the MS game, the referee uniformly randomly samples a pair of equations that intersect on a single variable, then sends one to Alice and the other to Bob. Like the MS game again, Alice and Bob each output three bits, winning if the parities of their outputs match that of the equations they were asked, and if their outputs are consistent on the variable that they have in common.

Formally, we define the doily game as follows:
\begin{definition}
The \emph{doily game} is defined as the nonlocal game $([15], \allowbreak [15], \allowbreak \{0, 1\}^3, \allowbreak \{0, 1\}^3, \allowbreak \pi, \allowbreak R)$, in which $\pi$ is the uniform distribution over
\begin{align}
    \{(j, k) \in [15] \times [15] : ~\text{$E_j$ and $E_k$ have exactly one variable in common}\}
\end{align}
and $R: [15] \times [15] \times \{0, 1\}^3 \times \{0, 1\}^3 \to \{0, 1\}$ is the function such that $R(j, k, a, b) = 1$ if and only if $a \in \{0, 1\}^3$ satisfies the parity of $E_j$, $b \in \{0, 1\}^3$ satisfies the parity of $E_k$, and the bit corresponding to the variable shared between $E_j$ and $E_k$ in $a$ is the same as that in $b$.

As with the MS game, we write $A_m(j)$ to denote Alice's bit output for $V_m$ when asked equation $E_j$, and the same for $B_m(j)$ and Bob.
\end{definition}

By our construction, the full set of 15 commuting triples of $P_2$ is naturally an operator solution to the doily game. Thus, just like the MS game, a perfect entangled strategy for the doily game is obtained by having the players share two Bell states and measure according to the operators of the commuting triples associated with the asked equations in \autoref{fig:ams_prods}.

However, the doily game on its own is insufficient to improve on the $\frac{1}{9}$ gap of the MS game, which we shall show in \autoref{sec:classical_doily}.

\subsection{Basic structure}

It turns out that the group structure of $P_2$ is intimately connected to finite projective geometry. To be precise, $P_2$ is isomorphic to $\mathbb{Z}_2^4$, which is in turn isomorphic to the projective space $PG(3, 2)$ \cite{marceaux2019}. However, recall that we are only interested in the subgroups of $P_2$ that correspond to commuting triples, as only those subgroups represent compatible measurements. As such, we may see in \autoref{fig:tet1} that our 15 commuting triples may be embedded as a substructure of the Fano tetrahedron commonly used to depict $PG(3, 2)$. Furthermore, we may see in \autoref{fig:tet2} that the original MS game appears as a sub-substructure of the Fano tetrahedron.

\iffigstwo
\begin{figure}[h]
    \centering
    \begin{tikzpicture}[line join=bevel,rotate around x=-40,rotate around z=-50,scale=0.75]
    \coordinate (V8) at (-3,-3,-3);
    \coordinate (V14) at (3,3,-3);
    \coordinate (V11) at (3,-3,3);
    \coordinate (V0) at (-3,3,3);
    \coordinate (V4) at (0,0,-3);
    \coordinate (V5) at (0,-3,0);
    \coordinate (V7) at (-3,0,0);
    \coordinate (V3) at (3,0,0);
    \coordinate (V13) at (0,3,0);
    \coordinate (V10) at (0,0,3);
    \coordinate (V2) at (1,-1,-1);
    \coordinate (V9) at (-1,1,-1);
    \coordinate (V12) at (-1,-1,1);
    \coordinate (V6) at (1,1,1);
    \coordinate (V1) at (0,0,0);
    
    \draw[draw=none,fill opacity=0.7,fill=black!30] (V8) -- (V14) -- (V11) -- cycle;
    
    \draw[line width=1pt] (V8) -- (V3);
    \draw[red,line width=1pt] (V8) -- (V13);
    \draw[line width=1pt] (V14) -- (V5);
    \draw[line width=1pt] (V14) -- (V7);
    \draw[line width=1pt] (V11) -- (V4);
    \draw[line width=1pt] (V0) -- (V4);
    
    \draw[line width=1.25pt] (V4) -- (V10);
    \draw[line width=1.25pt] (V5) -- (V13);
    \draw[line width=1.25pt] (V7) -- (V3);
    
    \fill (V4) circle (4pt);
    \fill (V2) circle (4pt);
    \fill (V9) circle (4pt);
    \fill (V1) circle (4pt);
    
    \draw[draw=none,fill opacity=0.75,fill=gray] (V8) -- (V11) -- (V0) -- cycle;
    \draw[draw=none,fill opacity=0.75,fill=gray!60] (V14) -- (V11) -- (V0) -- cycle;
    
    \draw[line width=1.5pt] (V8) -- (V10);
    \draw[red,line width=1.5pt] (V14) -- (V10);
    \draw[red,line width=1.5pt] (V11) -- (V7);
    \draw[line width=1.5pt] (V11) -- (V13);
    \draw[line width=1.5pt] (V0) -- (V5);
    \draw[line width=1.5pt] (V0) -- (V3);
    
    \fill (V8) circle (4pt);
    \fill (V14) circle (4pt);
    \fill (V11) circle (4pt);
    \fill (V0) circle (4pt);
    \fill (V5) circle (4pt);
    \fill (V7) circle (4pt);
    \fill (V3) circle (4pt);
    \fill (V13) circle (4pt);
    \fill (V10) circle (4pt);
    \fill (V12) circle (4pt);
    \fill (V6) circle (4pt);
    \end{tikzpicture}
    \caption{The 15 variables (points) and 15 equations (lines) of the doily game embedded in the Fano tetrahedron. 12 lines run along the surface while three lines cross the interior, connecting opposite edges through the centroid. The black lines have parity 0 while the red lines have parity 1.}
    \label{fig:tet1}
\end{figure}

\begin{figure}
    \centering
    \begin{tikzpicture}[line join=bevel,rotate around x=-40,rotate around z=-50,scale=0.75]
    \coordinate (V0) at (-3,-3,-3);
    \coordinate (V1) at (3,3,-3);
    \coordinate (V2) at (3,-3,3);
    \coordinate (V3) at (-3,3,3);
    \coordinate (V4) at (0,0,-3);
    \coordinate (V5) at (0,-3,0);
    \coordinate (V6) at (-3,0,0);
    \coordinate (V7) at (3,0,0);
    \coordinate (V8) at (0,3,0);
    \coordinate (V9) at (0,0,3);
    \coordinate (V10) at (1,-1,-1);
    \coordinate (V11) at (-1,1,-1);
    \coordinate (V12) at (-1,-1,1);
    \coordinate (V13) at (1,1,1);
    \coordinate (V14) at (0,0,0);
    
    \draw[draw=none,fill opacity=0.7,fill=black!30] (V0) -- (V1) -- (V2) -- cycle;
    
    \draw[red,line width=1pt] (V0) -- (V8);
    \draw[line width=1pt] (V1) -- (V6);
    
    \fill (V11) circle (4pt);
    
    \draw[draw=none,fill opacity=0.75,fill=gray] (V0) -- (V2) -- (V3) -- cycle;
    \draw[draw=none,fill opacity=0.75,fill=gray!60] (V1) -- (V2) -- (V3) -- cycle;
    
    \draw[line width=1.5pt] (V0) -- (V9);
    \draw[red,line width=1.5pt] (V1) -- (V9);
    \draw[red,line width=1.5pt] (V2) -- (V6);
    \draw[line width=1.5pt] (V2) -- (V8);
    
    \fill (V0) circle (4pt);
    \fill (V1) circle (4pt);
    \fill (V2) circle (4pt);
    \fill (V6) circle (4pt);
    \fill (V8) circle (4pt);
    \fill (V9) circle (4pt);
    \fill (V12) circle (4pt);
    \fill (V13) circle (4pt);
    \end{tikzpicture}
    \caption{The magic square game from \autoref{fig:ms_diagr} as a substructure of \autoref{fig:tet1}.}
    \label{fig:tet2}
\end{figure}
\fi

We formally state several observations about the commuting triples of $P_2$:

\begin{theorem} \label{thm:trip_props}
The following statements are all true:
\begin{enumerate}
    \item $P_2$ has 15 commuting triples.
    \item Every non-identity element of $P_2$ is a member of three commuting triples.
    \item For each commuting triple, there are exactly six other commuting triples with which it shares exactly one element.
    \item There are 90 ordered pairs of commuting triples that intersect on exactly one variable.
    \item $P_2$ has 10 operator magic squares.
\end{enumerate}
\end{theorem}
\begin{proof}
For statements 1 through 3, it is straightforward to check \autoref{fig:ams_prods} and \autoref{fig:tet1}. Statement 4 follows from 1 and 3. For statement 5, the operator magic square in \autoref{fig:ms_diagr} and the nine in \autoref{appx:squares} form the complete set.
\end{proof}

Using the above structure of the commuting triples, we can prove that the question distribution for the doily game may be achieved by four equivalent sampling procedures, which shall be useful in subsequent analysis of the doily game:
\begin{definition} \label{def:procs}
We define four sampling procedures:
\begin{enumerate}
    \item \emph{Flat}: The referee uniformly samples an ordered pair of intersecting equations to send to Alice and Bob.
    \item \emph{Equation-first}: The referee uniformly samples one equation to send to Alice, then samples one intersecting equation to send to Bob.
    \item \emph{Variable-first}: The referee uniformly samples a variable to test, then samples one equation containing the variable to send to Alice, and finally samples one of the other equations containing the variable to send to Bob.
    \item \emph{Magic-square-first}: The referee uniformly samples a set of six intersecting equations that forms one of the 10 MS-equivalent games, then samples an ordered pair of intersecting equations according to the MS-equivalent game's question distribution to send to Alice and Bob.
\end{enumerate}
\end{definition}
\begin{lemma}
The four sampling procedures in \Autoref{def:procs} are equivalent, and each ordered pair is selected with probability $\frac{1}{90}$.
\end{lemma}
\begin{proof}
The proof follows from \Autoref{thm:trip_props}.
\end{proof}

\subsection{Classical value} \label{sec:classical_doily}

Unfortunately, it turns out that the doily game by itself is not sufficient to widen the gap from the MS game, as the doily game also has classical value $\frac{8}{9}$. However, this setback shall be overcome by noticing a shared structure between all of the optimal classical strategies for the doily game.

\begin{theorem} \label{thm:ams_val}
The doily game has classical value $\frac{8}{9}$.
\end{theorem}
\begin{proof}
We prove the theorem in two parts: first, by showing that a success probability greater than $\frac{8}{9}$ is unattainable, and, second, by showing that a success probability of $\frac{8}{9}$ is attainable.
\begin{figure}[h]
    \centering
    \begin{tikzpicture}[scale=0.5]
    \coordinate (V1) at (90:3);
    \coordinate (V3) at (162:3);
    \coordinate (V6) at (234:3);
    \coordinate (V14) at (306:3);
    \coordinate (V5) at (18:3);
    
    \coordinate (V10) at (270:2);
    \coordinate (V2) at (342:2);
    \coordinate (V13) at (54:2);
    \coordinate (V7) at (126:2);
    \coordinate (V0) at (198:2);
    
    \coordinate (V4) at (270:5);
    \coordinate (V8) at (342:5);
    \coordinate (V11) at (54:5);
    \coordinate (V9) at (126:5);
    \coordinate (V12)  at (198:5);
    
    \draw[ultra thick] (V1) -- (V3);
    \draw[ultra thick] (V3) -- (V6);
    \draw[ultra thick] (V14) -- (V5);
    \draw[ultra thick] (V5) -- (V1);
    
    \draw[ultra thick] (V1) -- (V4);
    \draw[ultra thick] (V3) -- (V8);
    \draw[ultra thick] (V6) -- (V11);
    \draw[ultra thick] (V14) -- (V9);
    \draw[ultra thick] (V5) -- (V12);
    
    \draw[ultra thick] (V4) -- (V11);
    \draw[ultra thick] (V12) -- (V8);
    \draw[ultra thick] (V9) -- (V4);
    
    \draw[red,ultra thick] (V6) -- (V14);
    \draw[red,ultra thick] (V11) -- (V12);
    \draw[red,ultra thick] (V8) -- (V9);
    
    \draw[ultra thick,fill=white] (V1) circle [radius=0.75];
    \draw[ultra thick,fill=white] (V3) circle [radius=0.75];
    \draw[ultra thick,fill=white] (V6) circle [radius=0.75];
    \draw[ultra thick,fill=white] (V14) circle [radius=0.75];
    \draw[ultra thick,fill=white] (V5) circle [radius=0.75];
    \draw[ultra thick,fill=white] (V10) circle [radius=0.75];
    \draw[ultra thick,fill=white] (V2) circle [radius=0.75];
    \draw[ultra thick,fill=white] (V13) circle [radius=0.75];
    \draw[ultra thick,fill=white] (V7) circle [radius=0.75];
    \draw[ultra thick,fill=white] (V0) circle [radius=0.75];
    \draw[ultra thick,fill=white] (V4) circle [radius=0.75];
    \draw[ultra thick,fill=white] (V8) circle [radius=0.75];
    \draw[ultra thick,fill=white] (V11) circle [radius=0.75];
    \draw[ultra thick,fill=white] (V9) circle [radius=0.75];
    \draw[ultra thick,fill=white] (V12) circle [radius=0.75];
    
    \node at (V1) {\large 0};
    \node at (V3) {\large 0};
    \node at (V6) {\large 0};
    \path (V14) ++(180:0.48) node {\tiny 1};
    \path (V14) ++(126:0.48) node {\tiny 0};
    \path (V14) ++(72:0.48) node {\tiny 1};
    \path (V5) ++(144:0.48) node {\tiny 0};
    \path (V5) ++(198:0.48) node {\tiny 0};
    \path (V5) ++(252:0.48) node {\tiny 1};
    \node at (V10) {\large 0};
    \node at (V2) {\large 0};
    \path (V13) ++(54:0.48) node {\tiny 0};
    \node at (V13) {\tiny 1};
    \path (V13) ++(234:0.48) node {\tiny 1};
    \path (V7) ++(126:0.48) node {\tiny 0};
    \node at (V7) {\tiny 0};
    \path (V7) ++(306:0.48) node {\tiny 1};
    \node at (V0) {\large 0};
    \node at (V4) {\large 0};
    \node at (V8) {\large 0};
    \path (V11) ++(198:0.48) node {\tiny 0};
    \path (V11) ++(234:0.48) node {\tiny 1};
    \path (V11) ++(270:0.48) node {\tiny 0};
    \node at (V9) {\large 0};
    \node at (V12) {\large 0};
    \end{tikzpicture}
    \hspace{0.25cm}
    \begin{tikzpicture}[scale=0.5]
    \coordinate (V1) at (90:3);
    \coordinate (V3) at (162:3);
    \coordinate (V6) at (234:3);
    \coordinate (V14) at (306:3);
    \coordinate (V5) at (18:3);
    
    \coordinate (V10) at (270:2);
    \coordinate (V2) at (342:2);
    \coordinate (V13) at (54:2);
    \coordinate (V7) at (126:2);
    \coordinate (V0) at (198:2);
    
    \coordinate (V4) at (270:5);
    \coordinate (V8) at (342:5);
    \coordinate (V11) at (54:5);
    \coordinate (V9) at (126:5);
    \coordinate (V12)  at (198:5);
    
    \draw[ultra thick] (V1) -- (V3);
    \draw[ultra thick] (V3) -- (V6);
    \draw[ultra thick] (V14) -- (V5);
    \draw[ultra thick] (V5) -- (V1);
    
    \draw[ultra thick] (V1) -- (V4);
    \draw[ultra thick] (V3) -- (V8);
    \draw[ultra thick] (V6) -- (V11);
    \draw[ultra thick] (V14) -- (V9);
    \draw[ultra thick] (V5) -- (V12);
    
    \draw[ultra thick] (V4) -- (V11);
    \draw[ultra thick] (V12) -- (V8);
    \draw[ultra thick] (V9) -- (V4);
    
    \draw[red,ultra thick] (V6) -- (V14);
    \draw[red,ultra thick] (V11) -- (V12);
    \draw[red,ultra thick] (V8) -- (V9);
    
    \draw[ultra thick,fill=white] (V1) circle [radius=0.75];
    \draw[ultra thick,fill=white] (V3) circle [radius=0.75];
    \draw[ultra thick,fill=white] (V6) circle [radius=0.75];
    \draw[ultra thick,fill=white] (V14) circle [radius=0.75];
    \draw[ultra thick,fill=white] (V5) circle [radius=0.75];
    \draw[ultra thick,fill=white] (V10) circle [radius=0.75];
    \draw[ultra thick,fill=white] (V2) circle [radius=0.75];
    \draw[ultra thick,fill=white] (V13) circle [radius=0.75];
    \draw[ultra thick,fill=white] (V7) circle [radius=0.75];
    \draw[ultra thick,fill=white] (V0) circle [radius=0.75];
    \draw[ultra thick,fill=white] (V4) circle [radius=0.75];
    \draw[ultra thick,fill=white] (V8) circle [radius=0.75];
    \draw[ultra thick,fill=white] (V11) circle [radius=0.75];
    \draw[ultra thick,fill=white] (V9) circle [radius=0.75];
    \draw[ultra thick,fill=white] (V12) circle [radius=0.75];
    
    \node at (V1) {\large 0};
    \node at (V3) {\large 0};
    \node at (V6) {\large 0};
    \path (V14) ++(180:0.48) node {\tiny 1};
    \path (V14) ++(126:0.48) node {\tiny 1};
    \path (V14) ++(72:0.48) node {\tiny 0};
    \path (V5) ++(144:0.48) node {\tiny 1};
    \path (V5) ++(198:0.48) node {\tiny 0};
    \path (V5) ++(252:0.48) node {\tiny 0};
    \node at (V10) {\large 0};
    \node at (V2) {\large 0};
    \path (V13) ++(54:0.48) node {\tiny 1};
    \node at (V13) {\tiny 0};
    \path (V13) ++(234:0.48) node {\tiny 1};
    \path (V7) ++(126:0.48) node {\tiny 0};
    \node at (V7) {\tiny 1};
    \path (V7) ++(306:0.48) node {\tiny 0};
    \node at (V0) {\large 0};
    \node at (V4) {\large 0};
    \node at (V8) {\large 0};
    \path (V11) ++(198:0.48) node {\tiny 1};
    \path (V11) ++(234:0.48) node {\tiny 0};
    \path (V11) ++(270:0.48) node {\tiny 0};
    \node at (V9) {\large 0};
    \node at (V12) {\large 0};
    \end{tikzpicture}
    \caption{Planar graphs representing the classical strategy for the doily game that wins with probability $\frac{8}{9}$. Alice's strategy is on the left and Bob's strategy is on the right. Each small number is positioned along the equation it is output for, with the exception of the centered small numbers being output for the equation passing through the circle's center.}
    \label{fig:classical_doily}
\end{figure}
\begin{table}[h]
    \centering
    {
    \NiceMatrixOptions{columns-width=0.15cm,cell-space-limits=0.15cm}
    \begin{align*}
    \begin{NiceArray}{c|ccc}
        E_0 & \sta{V_0} & \sta{V_3} & \sta{V_6} \\
        E_1 & \sta{V_1} & \sta{V_3} & \sta{V_7} \\
        E_2 & \sta{V_2} & \sta{V_3} & \sta{V_8} \\
        E_3 & \sta{V_0} & \sta{V_4} & \sta{V_9} \\
        E_4 & \sta{V_1} & \sta{V_4} & \sta{V_{10}} \\
        E_5 & \sta{V_2} & \sta{V_4} & \sta{V_{11}} \\
        E_6 & \sta{V_0} & \sta{V_5} & \sta{V_{12}} \\
        E_7 & \sta{V_1} & \sta{V_5} & \sta{V_{13}} \\
        E_8 & \sta{V_2} & \stb{V_5} & \stb{V_{14}} \\
        E_9 & \sta{V_6} & \stb{V_{11}} & \stb{V_{13}} \\
        E_{10} & \sta{V_{10}} & \sta{V_{12}} & \sta{V_8} \\
        E_{11} & \sta{V_{14}} & \sta{V_7} & \sta{V_9} \\
        E_{12} & \sta{V_6} & \sta{V_{10}} & \stb{V_{14}} \\
        E_{13} & \sta{V_{11}} & \sta{V_{12}} & \stb{V_7} \\
        E_{14} & \stb{V_{13}} & \sta{V_8} & \sta{V_9} \\
    \end{NiceArray} \hspace{2.5cm}
    \begin{NiceArray}{c|ccc}
        E_0 & \sta{V_0} & \sta{V_3} & \sta{V_6} \\
        E_1 & \sta{V_1} & \sta{V_3} & \sta{V_7} \\
        E_2 & \sta{V_2} & \sta{V_3} & \sta{V_8} \\
        E_3 & \sta{V_0} & \sta{V_4} & \sta{V_9} \\
        E_4 & \sta{V_1} & \sta{V_4} & \sta{V_{10}} \\
        E_5 & \sta{V_2} & \sta{V_4} & \sta{V_{11}} \\
        E_6 & \sta{V_0} & \sta{V_5} & \sta{V_{12}} \\
        E_7 & \sta{V_1} & \stb{V_5} & \stb{V_{13}} \\
        E_8 & \sta{V_2} & \sta{V_5} & \sta{V_{14}} \\
        E_9 & \sta{V_6} & \sta{V_{11}} & \sta{V_{13}} \\
        E_{10} & \sta{V_{10}} & \sta{V_{12}} & \sta{V_8} \\
        E_{11} & \stb{V_{14}} & \stb{V_7} & \sta{V_9} \\
        E_{12} & \sta{V_6} & \sta{V_{10}} & \stb{V_{14}} \\
        E_{13} & \stb{V_{11}} & \sta{V_{12}} & \sta{V_7} \\
        E_{14} & \stb{V_{13}} & \sta{V_8} & \sta{V_9} \\
    \end{NiceArray}
    \end{align*}
    }%
    \caption{Tables representing the classical strategy for the doily game that wins with probability $\frac{8}{9}$. Alice's strategy is on the left and Bob's strategy is on the right. A contoured white variable label indicates that the variable is output as 0 for that equation while a solid black variable label indicates that the variable is output as 1 for that equation.}
    \label{tab:classical_doily}
\end{table}
\begin{enumerate}
    \item Assume towards contradiction that the strategy $S$ wins the doily game with probability $p > \frac{8}{9}$. By the magic-square-first procedure from \Autoref{def:procs}, there must be some MS-equivalent game that $S$ wins with probability greater than $\frac{8}{9}$. This contradicts the fact that the classical values of the MS-equivalent games are all $\frac{8}{9}$, proving that the assumption that there exists a classical strategy that wins the doily game with probability greater than $\frac{8}{9}$ is false.
    \item We provide an example of a classical strategy that wins with probability $\frac{8}{9}$ as diagrams in \autoref{fig:classical_doily} and as tables in \autoref{tab:classical_doily}, found via computational search. Since the players only lose when one player outputs 1 for a variable that the other outputs as 0, we may simply count that there are 10 of these cases in total. Following statement 4 of \Autoref{thm:trip_props}, there are 90 possible pairs in total, so the players win with probability $\frac{8}{9}$.
\end{enumerate}
\end{proof}

\subsection{Additional properties}

As mentioned above, it turns out that all optimal classical strategies for the doily game exhibit a certain property, which can be exploited to further augment the doily game into a third game that ultimately has a classical value less than $\frac{8}{9}$. In particular, all optimal classical strategies for the doily game are \emph{asymmetric}, which we define below.

\begin{definition}
A classical strategy $S = (A, B)$ is said to be \emph{symmetric}\footnote{Not to be confused with the term \emph{synchronous}, which will also be used later.} if $A = B$. Otherwise, $S$ is said to be \emph{asymmetric}.
\end{definition}

To show that all optimal doily game strategies are asymmetric, we find that the best symmetric strategies win with probability less than $\frac{8}{9}$.

\begin{theorem} \label{thm:opt_sym}
An optimal symmetric strategy for the doily game wins with probability $\frac{13}{15} < \frac{8}{9}$.
\end{theorem}

We follow the argument of \cite{Kelleher2024} using the language developed in this work. We first consider the form of symmetric strategies for the doily game.

Let $S$ be a symmetric strategy for the doily game and $m$ denote the index of some variable. There are only two possible cases for how $S$ responds with the assignment of $V_m$: that all three equations containing $V_m$ assign the same bit to $V_m$, and that two of the equations assign the same bit to $V_m$ while the third assigns the opposite bit to $V_m$. As these two cases are quite important, we refer to them as follows:

\begin{definition} \label{def:consist}
When a symmetric strategy $S$ for the doily game outputs the same bit for variable $V_m$ for each of the three equations that include $V_m$, we refer to $S$ having a \emph{consistent assignment} for $V_m$. Otherwise, we refer to $S$ having a $\frac{2}{3}$-\emph{consistent assignment} for $V_m$.
\end{definition}

Breaking down symmetric strategies in this way yields a simple expression for the winning probabilities of doily game strategies that satisfy all of the game's parity constraints.

\begin{lemma} \label{thm:sym_prob}
Let $S$ be a symmetric strategy for the doily game that satisfies all parity constraints with $q$ variables having consistent assignments and $(15 - q)$ variables having $\frac{2}{3}$-consistent assignments. Then, $S$ wins the doily game with probability $\frac{15+2q}{45}$.
\end{lemma}
\begin{proof}
By the variable-first procedure from \Autoref{def:procs}, the referee tests a variable that has a consistent assignment according to $S$ with probability $\frac{q}{15}$ and tests a variable that has a $\frac{2}{3}$-consistent assignment according to $S$ with probability $\frac{15-q}{15}$.
\begin{enumerate}
    \item Suppose that the referee tests a variable with a consistent assignment. By the left-hand side of \autoref{fig:consist}, $S$ always passes the consistency check, so the players win with probability 1.
    \item Suppose that the referee tests a variable with a $\frac{2}{3}$-consistent assignment. Without loss of generality, the players' strategy is of the form shown on the right-hand side of \autoref{fig:consist}. Since the referee may ask each of the six equation pairs with equal probability, the players win with probability $\frac{1}{3}$.
\end{enumerate}
\iffigsone
\begin{figure}[h]
    \centering
    \begin{tikzpicture}[scale=1.75]
    \node (Vm) at (0,1.25) {\Large $V_m$ is consistent};
    
    \node (Ej) at (-1,0.5) {\Large $E_j$};
    \node (Ek) at (0,0.5) {\Large $E_k$};
    \node (El) at (1,0.5) {\Large $E_\ell$};
    
    \node (A) at (-1.5,0) {\Large $A$};
    \node (B) at (-1.5,-1) {\Large $B$};
    
    \node (Aj) at (-1,0) {$b$};
    \node (Ak) at (0,0) {$b$};
    \node (Al) at (1,0) {$b$};
    \node (Bj) at (-1,-1) {$b$};
    \node (Bk) at (0,-1) {$b$};
    \node (Bl) at (1,-1) {$b$};

    \draw[very thick,blue] (Aj) -- (Bk);
    \draw[very thick,blue] (Aj) -- (Bl);
    \draw[very thick,blue] (Ak) -- (Bj);
    \draw[very thick,blue] (Ak) -- (Bl);
    \draw[very thick,blue] (Al) -- (Bj);
    \draw[very thick,blue] (Al) -- (Bk);
    \end{tikzpicture}
    ~~~
    \begin{tikzpicture}[scale=1.75]
    \node (Vm) at (0,1.25) {\Large $V_m$ is $\frac{2}{3}$-consistent};
    
    \node (Ej) at (-1,0.5) {\Large $E_j$};
    \node (Ek) at (0,0.5) {\Large $E_k$};
    \node (El) at (1,0.5) {\Large $E_\ell$};
    
    \node (A) at (-1.5,0) {\Large $A$};
    \node (B) at (-1.5,-1) {\Large $B$};
    
    \node (Aj) at (-1,0) {$b$};
    \node (Ak) at (0,0) {$\overline{b}$};
    \node (Al) at (1,0) {$\overline{b}$};
    \node (Bj) at (-1,-1) {$b$};
    \node (Bk) at (0,-1) {$\overline{b}$};
    \node (Bl) at (1,-1) {$\overline{b}$};

    \draw[very thick,red] (Aj) -- (Bk);
    \draw[very thick,red] (Aj) -- (Bl);
    \draw[very thick,red] (Ak) -- (Bj);
    \draw[very thick,red] (Al) -- (Bj);
    \draw[very thick,blue] (Ak) -- (Bl);
    \draw[very thick,blue] (Al) -- (Bk);
    \end{tikzpicture}
    \caption{Alice and Bob's possible responses for variable $V_m$ when their strategy gives it a consistent assignment (left) and $\frac{2}{3}$-consistent assignment (right). Each line represents a pair of equations intersecting on $V_m$ that either passes (blue) or fails (red) the consistency check.}
    \label{fig:consist}
\end{figure}
\fi

Thus, the unconditional probability that the players win the doily game using $S$ is given by
\begin{align}
    \frac{q}{15} \cdot 1 + \frac{15-q}{15} \cdot \frac{1}{3} = \frac{15+2q}{45}.
\end{align}
\end{proof}

Additionally, we use computational exhaustive search to determine the maximal number of the 15 doily game equations that we may satisfy by assigning a fixed bit to each variable:

\begin{theorem} \label{thm:eqn_sat}
A bit assignment $b: [15] \to \{0, 1\}$ may satisfy at most 12 out of the 15 doily game equations.
\end{theorem}
\begin{proof}
The theorem is proven by computational exhaustive search using the program in \autoref{appx:sat}.
\end{proof}

Once we are equipped with the above tools, we may prove that $\frac{13}{15}$ is indeed the maximal success probability for a symmetric strategy for the doily game:

\begin{proof}[Proof of \Autoref{thm:opt_sym}]
Our proof of the theorem is in two parts: first, we show that a success probability greater than $\frac{13}{15}$ is unattainable, and, second, we show that a success probability of $\frac{13}{15}$ is attainable.

\begin{enumerate}
    \item We show that a success probability greater than $\frac{13}{15}$ is unattainable by contradiction. Assume that $S$ is an optimal symmetric strategy for the doily game with success probability greater than $\frac{13}{15}$.
    
    Assume without loss of generality that $S$ satisfies all parity constraints. We may make this assumption since $S$ is optimal, and any symmetric strategy that violates parity constraints can only be improved by flipping an arbitrary bit for both players in each violating response.
    
    By \Autoref{thm:sym_prob}, $S$ must have at most two $\frac{2}{3}$-consistent assignments to have a success probability greater than $\frac{13}{15}$. For each variable index $m \in [15]$, let the bit assignment $b(m) \in \{0, 1\}$ be the majority value of the possible bit outputs for $V_m$ according to $S$. Because $S$ satisfies all parity constraints, and $S$ only has two $\frac{2}{3}$-consistent variables, the bit assignment $b$ must satisfy at least 13 out of the 15 equations. However, this contradicts \Autoref{thm:eqn_sat}, so the assumption that $S$ is a symmetric strategy for the doily game with with success probability greater than $\frac{13}{15}$ is false.
    
    \item Towards showing that a success probability of $\frac{13}{15}$ is attainable, let $b(m) \in \{0, 1\}$ for $m \in [15]$ be a bit assignment that satisfies 12 out of the 15 equations as in \Autoref{thm:eqn_sat}. Let $M \subseteq [15]$ be a set of three indices such that $V_m$ is in a different unsatisfied equation for each $m \in M$. Let $S$ be the symmetric strategy in which $V_m$ for $m \in M$ has the $\frac{2}{3}$-consistent assignment that is $b(m)$ except for when the equation is one of the unsatisfied equations, and $V_m$ for any other $r$ has the consistent assignment $b(m)$. Such a strategy satisfies all parity constraints, and, by \Autoref{thm:sym_prob}, $S$ wins with probability $\frac{13}{15}$.
\end{enumerate}
\end{proof}

In the above proof, we reason that there must exist an optimal classical strategy that satisfies all of the game's parity constraints. It will be useful for later to note that \emph{all} optimal classical strategies for the doily game satisfy all of the game's parity constraints.

\begin{theorem} \label{thm:opt_par}
Let $S$ be an optimal classical strategy for the doily game. $S$ satisfies all of the game's parity constraints.
\end{theorem}
\begin{proof}
We prove this theorem by contradiction. In particular, we show that an optimal doily strategy that violates a parity constraint can be improved by flipping one of the bits in the violated equation.

Assume towards contradiction that an optimal doily strategy $S = (A, B)$ violates the parity constraint on Alice's side of equation $E_j$, containing variables $V_m$, $V_n$, and some other third variable. Let the strategy $S' = (A', B')$ be defined such that $B' = B$ and $A'$ is almost identical to $A$, except that the bit corresponding to $V_m$ in Alice's response to $E_j$ is flipped, fixing the parity violation.

By the equation first procedure from \Autoref{def:procs}, $S$ wins with probability 0 conditioned on the referee sending Alice equation $E_j$ due to the parity violation, and $S$ wins with some probability $p$ conditioned on the referee sending Alice any equation other than $E_j$. Because $S'$ is identical to $S$ except for $A_m(j) \neq A'_m(j)$, $S'$ also wins with probability $p$ conditioned on the referee sending Alice any equation other than $E_j$.

As $S$ is optimal, $S'$ cannot be better than $S$. Thus, $S'$ also wins with probability 0 conditioned on the referee sending Alice equation $E_j$. Moreover, since $S'$ satisfies the parity constraint of $E_j$, it must be the case that $S'$ never satisfies the consistency constraints when Alice receives $E_j$.

Let $E_k$ denote one of the two equations other than $E_j$ that contains $V_n$. Let
\begin{align}
    b = A'_n(j).
\end{align}
As $S'$ never satisfies the consistency constraints when Alice receives $E_j$,
\begin{align}
    B'_n(k) = \overline{b}.
\end{align}

Now, consider the strategy $S'' = (A'', B'')$, defined such that $B'' = B$ and $A''$ is again almost identical to $A$, except that the bit corresponding to $V_n$ in Alice's response to $E_j$ is flipped. Then, we have that
\begin{align}
    A''_n(j) &= \overline{A_n(j)} \\
    &= \overline{A'_n(j)} \\
    &= \overline{b} \\
    &= B'_n(k) \\
    &= B''_n(k).
\end{align}
Thus, $S''$ wins when the equation-first procedure selects $E_j$ to send to Alice, then selects $E_k$ to send to Bob. However, as with $S'$, $S''$ must win with probability 0 conditioned on the referee sending Alice equation $E_j$, creating a contradiction. Thus, contradiction proves that the assumption that $S$ violates a parity constraint is false.
\end{proof}

\section{The \texorpdfstring{$p$}{p}-synchronous (\texorpdfstring{$p$}{p}-sync) doily game} \label{sec:psams}

Because any optimal classical strategy for the doily game must be asymmetric, we design a new game that is a probabilistic mixture of the doily game and a second game that penalizes asymmetric strategies. All of the questions from the doily game are retained, but we introduce 15 new \emph{synchronous} questions: with probability $p$, instead of playing the doily game, the referee uniformly randomly samples a single equation out of the 15 and sends it to both Alice and Bob, requiring that their answers are completely identical while still requiring that the parity constraint of that equation is satisfied. We call this game the \emph{$p$-synchronous doily game}.

Since asymmetric strategies disagree between Alice and Bob on the output for at least one equation $E_j$, the players lose if they are synchronously both asked $E_j$. As desired, an asymmetric strategy cannot win the $p$-synchronous doily game with probability 1.\footnote{Interestingly, the same trick does not work for the MS game, as the MS game already has optimal symmetric strategies. As such, adding synchronous questions only benefits the players by giving them free wins.}

We formally define the game as follows:
\begin{definition}
For $p \in [0, 1]$, the \emph{$p$-synchronous ($p$-sync) doily game} is defined as the nonlocal game $([15], [15], \{0, 1\}^3, \{0, 1\}^3, (1-p)\pi_1 + p\pi_2, R)$, for $\pi_1$ being the uniform distribution over
\begin{align}
    \{(j, k) \in [15] \times [15] : ~\text{$E_j$ and $E_k$ have exactly one variable in common}\},
\end{align}
$\pi_2$ being the uniform distribution over
\begin{align}
    \{(j, j) \in [15] \times [15] : j \in [15]\},
\end{align}
and $R: [15] \times [15] \times \{0, 1\}^3 \times \{0, 1\}^3 \to \{0, 1\}$ being the function such that $R(j, k, a, b) = 1$ if and only if $a$ satisfies the parity of $E_j$, $b$ satisfies the parity of $E_k$ and the bits corresponding to any variables shared between $E_j$ and $E_k$ in $a$ are the same as those in $b$.

As with the prior two games, we write $A_m(j)$ to denote Alice's bit output for $V_m$ when asked equation $E_j$, and the same for $B_m(j)$ and Bob.
\end{definition}

From the definition, it is clear to see that the perfect entangled strategy for the doily game is also a perfect entangled strategy for the $p$-sync doily game for all values of $p$. Specifically, the 0-sync doily game is just the doily game, and the 1-sync doily game is a trivial game in which Alice's equation is always the same as Bob's. These have classical values $\frac{8}{9}$ and $1$, respectively. However, the $\frac{1}{7}$-sync doily game has classical value $\frac{31}{35}$, lower than that of either the MS or doily games.

\begin{theorem} \label{thm:psams}
The $\frac{1}{7}$-sync doily game has classical value $\frac{31}{35}$.
\end{theorem}
To prove the above theorem, we will use one main lemma:
\begin{lemma} \label{lem:sync}
Let $S$ be an optimal asymmetric strategy for the doily game. $S$ wins on at most 13 of the 15 $p$-sync doily synchronous questions.
\end{lemma}
\begin{proof}
At a high level, we prove this lemma by contradiction through an analysis of cases. We first show that we need only consider when $S$ wins exactly 14 of the synchronous questions. We next consider the possible strategies obtained by symmetrizing $S$, and show that there are only nine possible cases. Of those nine cases, we then show that there are only five possible probabilities of winning, each of which we prove leads to a contradiction.

Assume towards contradiction that $S = (A, B)$ wins on more than 13 of the 15 synchronous questions. By \Autoref{thm:opt_par}, $S$ satisfies all of the game's parity constraints. As mentioned previously, $S$ cannot both be an asymmetric strategy and win on all 15 out of 15 synchronous questions. Thus, $S$ wins on exactly 14 of the 15 synchronous questions. As such, there exists exactly one equation $E_j$ such that
\begin{align}
    A(j) \neq B(j).
\end{align}
Moreover, as $A$ and $B$ both satisfy all of the parity constraints, $A(j)$ and $B(j)$ differ by exactly two bit flips. Let $V_m$ and $V_n$ denote the two variables such that
\begin{align}
    A_m(j) &\neq B_m(j) & A_n(j) &\neq B_n(j).
\end{align}

Now, let $S' = (A', B')$ denote the symmetric strategy defined by symmetrizing $S$ so that Alice's and Bob's local strategies under $S'$ are both the same as Alice's local strategy under $S$. In other words, we have that
\begin{align}
    A' = B' = A.
\end{align}

Consider the difference in success probabilities between $S$ and $S'$ conditioned on the variable-first procedure from \Autoref{def:procs} selecting $V_m$. Recall from \Autoref{def:consist} that there are two possible cases: that $S'$ has a consistent assignment for $V_m$ and that $S'$ has a $\frac{2}{3}$-consistent assignment for $V_m$.
\begin{enumerate}
    \item \label{case:1} Assume that $S'$ has a consistent assignment for $V_m$. Let $k$ and $\ell$ denote the indices of the two equations other than equation $j$ that involve variable $V_m$. Then, Alice and Bob both always output the same bit value, $b$, for $V_m$ when asked any of the equations $E_j$, $E_k$, and $E_\ell$. Thus,
    \begin{align}
        A'_m(j) &= b, & A'_m(k) &= b, & A'_m(\ell) &= b, \\
        B'_m(j) &= b, & B'_m(k) &= b, & B'_m(\ell) &= b. \label{eqn:case_1_1}
    \end{align}
    As with \autoref{fig:consist} in the proof of \Autoref{thm:sym_prob}, we count the cases in which Alice and Bob agree or disagree on the value of $V_m$ under $S'$ on the right-hand side of \autoref{fig:case_1} (the left-hand side will be discussed shortly). From the figure, we see that $S'$ wins the doily game with probability $1$, conditioned on the variable-first procedure selecting $V_m$.

    \iffigsone
    \begin{figure}[t]
        \centering
        \begin{tikzpicture}[scale=1.75]
        \node (Vm) at (0,1.125) {\Large $V_m$ under $S$};
        
        \node (Ej) at (-1,0.5) {\Large $E_j$};
        \node (Ek) at (0,0.5) {\Large $E_k$};
        \node (El) at (1,0.5) {\Large $E_\ell$};
        
        \node (A) at (-1.5,0) {\Large $A$};
        \node (B) at (-1.5,-1) {\Large $B$};
        
        \node (Aj) at (-1,0) {$b$};
        \node (Ak) at (0,0) {$b$};
        \node (Al) at (1,0) {$b$};
        \node (Bj) at (-1,-1) {$\overline{b}$};
        \node (Bk) at (0,-1) {$b$};
        \node (Bl) at (1,-1) {$b$};
    
        \draw[very thick,red] (Ak) -- (Bj);
        \draw[very thick,red] (Al) -- (Bj);
        \draw[very thick,blue] (Aj) -- (Bk);
        \draw[very thick,blue] (Aj) -- (Bl);
        \draw[very thick,blue] (Ak) -- (Bl);
        \draw[very thick,blue] (Al) -- (Bk);
        \end{tikzpicture}
        ~~~
        \begin{tikzpicture}[scale=1.75]
        \node (Vm) at (0,1.125) {\Large $V_m$ under $S'$};
        
        \node (Ej) at (-1,0.5) {\Large $E_j$};
        \node (Ek) at (0,0.5) {\Large $E_k$};
        \node (El) at (1,0.5) {\Large $E_\ell$};
        
        \node (A) at (-1.5,0) {\Large $A'$};
        \node (B) at (-1.5,-1) {\Large $B'$};
        
        \node (Aj) at (-1,0) {$b$};
        \node (Ak) at (0,0) {$b$};
        \node (Al) at (1,0) {$b$};
        \node (Bj) at (-1,-1) {$b$};
        \node (Bk) at (0,-1) {$b$};
        \node (Bl) at (1,-1) {$b$};
    
        \draw[very thick,blue] (Aj) -- (Bk);
        \draw[very thick,blue] (Aj) -- (Bl);
        \draw[very thick,blue] (Ak) -- (Bj);
        \draw[very thick,blue] (Ak) -- (Bl);
        \draw[very thick,blue] (Al) -- (Bj);
        \draw[very thick,blue] (Al) -- (Bk);
        \end{tikzpicture}
        \caption{Alice and Bob's possible responses for variable $V_m$ according to $S$ and according to $S'$ in \autoref{case:1}. Line color indicates whether the players pass (blue) or fail (red) the consistency check.}
        \label{fig:case_1}
    \end{figure}
    \fi
    
    Now, we may return to considering $S$. As $E_j$ is the equation on which $B$ differs from $B'$ on the bit assignment of $V_m$, it is the case that
    \begin{align}
        A_m(j) &= b, & A_m(k) &= b, & A_m(\ell) &= b, \\
        B_m(j) &= \overline{b}, & B_m(k) &= b, & B_m(\ell) &= b, \label{eqn:case_1_2}
    \end{align}
    which is also depicted on the left-hand side of \autoref{fig:case_1}. From the figure, we see that $S$ wins the doily game with probability $\frac{2}{3}$, conditioned on the variable-first procedure selecting $V_m$. Thus, the change in conditional success probability when switching from $S$ to $S'$ is $\frac{1}{3}$.

    \item Assume that $S'$ has a $\frac{2}{3}$-consistent assignment for $V_m$. Then, there are another two possible cases: that $E_j$ is one of the two equations that agree on the assignment of $V_m$ and that $E_j$ is the equation that differs from the other two on the assignment of $V_m$.
    \begin{enumerate}
        \item \label{case:2a} Assume that $E_j$ is one of the two equations that agree on the assignment of $V_m$. Let $k$ and $\ell$ denote the indices of the two equations other than equation $j$ that involve variable $V_m$.
        \iffigsone
        \begin{figure}[b]
            \centering
            \begin{tikzpicture}[scale=1.75]
            \node (Vm) at (0,1.125) {\Large $V_m$ under $S$};
            
            \node (Ej) at (-1,0.5) {\Large $E_j$};
            \node (Ek) at (0,0.5) {\Large $E_k$};
            \node (El) at (1,0.5) {\Large $E_\ell$};
            
            \node (A) at (-1.5,0) {\Large $A$};
            \node (B) at (-1.5,-1) {\Large $B$};
            
            \node (Aj) at (-1,0) {$b$};
            \node (Ak) at (0,0) {$b$};
            \node (Al) at (1,0) {$\overline{b}$};
            \node (Bj) at (-1,-1) {$\overline{b}$};
            \node (Bk) at (0,-1) {$b$};
            \node (Bl) at (1,-1) {$\overline{b}$};
        
            \draw[very thick,red] (Aj) -- (Bl);
            \draw[very thick,red] (Ak) -- (Bj);
            \draw[very thick,red] (Ak) -- (Bl);
            \draw[very thick,red] (Al) -- (Bk);
            \draw[very thick,blue] (Aj) -- (Bk);
            \draw[very thick,blue] (Al) -- (Bj);
            \end{tikzpicture}
            ~~~
            \begin{tikzpicture}[scale=1.75]
            \node (Vm) at (0,1.125) {\Large $V_m$ under $S'$};
            
            \node (Ej) at (-1,0.5) {\Large $E_j$};
            \node (Ek) at (0,0.5) {\Large $E_k$};
            \node (El) at (1,0.5) {\Large $E_\ell$};
            
            \node (A) at (-1.5,0) {\Large $A'$};
            \node (B) at (-1.5,-1) {\Large $B'$};
            
            \node (Aj) at (-1,0) {$b$};
            \node (Ak) at (0,0) {$b$};
            \node (Al) at (1,0) {$\overline{b}$};
            \node (Bj) at (-1,-1) {$b$};
            \node (Bk) at (0,-1) {$b$};
            \node (Bl) at (1,-1) {$\overline{b}$};
        
            \draw[very thick,red] (Aj) -- (Bl);
            \draw[very thick,red] (Ak) -- (Bl);
            \draw[very thick,red] (Al) -- (Bj);
            \draw[very thick,red] (Al) -- (Bk);
            \draw[very thick,blue] (Aj) -- (Bk);
            \draw[very thick,blue] (Ak) -- (Bj);
            \end{tikzpicture}
            \caption{Alice and Bob's possible responses for variable $V_m$ according to $S$ and according to $S'$ in \autoref{case:2a}. Line color indicates whether the players pass (blue) or fail (red) the consistency check.}
            \label{fig:case_2a}
        \end{figure}
        \fi
        Without loss of generality, Alice and Bob output the same bit value $b$ for $V_m$ when asked either $E_j$ or $E_k$, and output the opposite bit value $\overline{b}$ when asked $E_\ell$. As with \autoref{case:1}, we may depict Alice and Bob's possible responses conditioned on the variable-first procedure selecting $V_m$ according to $S'$ on the right-hand side of \autoref{fig:case_2a}, then derive their possible responses according to $S$ on the left-hand side.
        
        By counting the winning and losing responses in the figure, we see that $S'$ wins the doily game with conditional probability $\frac{1}{3}$, and $S$ also wins with conditional probability $\frac{1}{3}$. Thus, the change in conditional success probability when switching from $S$ to $S'$ is $0$.
    
        \item \label{case:2b} Assume that $E_j$ is the equation that differs from the other two on the assignment of $V_m$. Let $k$ and $\ell$ denote the indices of the two equations other than equation $j$ that involve variable $V_m$. Then, Alice and Bob both output the bit value $b$ for $V_m$ when asked $E_j$, but output the opposite bit value $\overline{b}$ when asked either $E_k$ or $E_\ell$. As with \autoref{case:1}, we may depict Alice and Bob's possible responses conditioned on the variable-first procedure selecting $V_m$ according to $S'$ on the right-hand side of \autoref{fig:case_2b}, then derive their possible responses according to $S$ on the left-hand side.

        \iffigsone
        \begin{figure}[b]
            \centering
            \begin{tikzpicture}[scale=1.75]
            \node (Vm) at (0,1.125) {\Large $V_m$ under $S$};
            
            \node (Ej) at (-1,0.5) {\Large $E_j$};
            \node (Ek) at (0,0.5) {\Large $E_k$};
            \node (El) at (1,0.5) {\Large $E_\ell$};
            
            \node (A) at (-1.5,0) {\Large $A$};
            \node (B) at (-1.5,-1) {\Large $B$};
            
            \node (Aj) at (-1,0) {$b$};
            \node (Ak) at (0,0) {$\overline{b}$};
            \node (Al) at (1,0) {$\overline{b}$};
            \node (Bj) at (-1,-1) {$\overline{b}$};
            \node (Bk) at (0,-1) {$\overline{b}$};
            \node (Bl) at (1,-1) {$\overline{b}$};
        
            \draw[very thick,red] (Aj) -- (Bk);
            \draw[very thick,red] (Aj) -- (Bl);
            \draw[very thick,blue] (Ak) -- (Bj);
            \draw[very thick,blue] (Ak) -- (Bl);
            \draw[very thick,blue] (Al) -- (Bj);
            \draw[very thick,blue] (Al) -- (Bk);
            \end{tikzpicture}
            ~~~
            \begin{tikzpicture}[scale=1.75]
            \node (Vm) at (0,1.125) {\Large $V_m$ under $S'$};
            
            \node (Ej) at (-1,0.5) {\Large $E_j$};
            \node (Ek) at (0,0.5) {\Large $E_k$};
            \node (El) at (1,0.5) {\Large $E_\ell$};
            
            \node (A) at (-1.5,0) {\Large $A'$};
            \node (B) at (-1.5,-1) {\Large $B'$};
            
            \node (Aj) at (-1,0) {$b$};
            \node (Ak) at (0,0) {$\overline{b}$};
            \node (Al) at (1,0) {$\overline{b}$};
            \node (Bj) at (-1,-1) {$b$};
            \node (Bk) at (0,-1) {$\overline{b}$};
            \node (Bl) at (1,-1) {$\overline{b}$};
        
            \draw[very thick,red] (Aj) -- (Bk);
            \draw[very thick,red] (Aj) -- (Bl);
            \draw[very thick,red] (Ak) -- (Bj);
            \draw[very thick,red] (Al) -- (Bj);
            \draw[very thick,blue] (Ak) -- (Bl);
            \draw[very thick,blue] (Al) -- (Bk);
            \end{tikzpicture}
            \caption{Alice and Bob's possible responses for variable $V_m$ according to $S$ and according to $S'$ in \autoref{case:2b}. Line color indicates whether the players pass (blue) or fail (red) the consistency check.}
            \label{fig:case_2b}
        \end{figure}
        \fi

        By counting the winning and losing responses in the figure, we see that $S'$ still wins the doily game with conditional probability $\frac{1}{3}$, but $S$ wins with conditional probability $\frac{2}{3}$. Thus, the change in conditional success probability when switching from $S$ to $S'$ is $-\frac{1}{3}$.
    \end{enumerate}
\end{enumerate}
In essence, conditioned on the variable-first procedure from \Autoref{def:procs} selecting $V_m$, switching strategy from $S$ to $S'$ changes the success probability by $\frac{1}{3}$ if $V_m$ satisfies \autoref{case:1}, 0 if $V_m$ satisfies \autoref{case:2a}, and $-\frac{1}{3}$ if $V_m$ satisfies \autoref{case:2b}. The same logic for $V_m$ also applies to $V_n$, so there are nine possible ways in total that $S'$ may differ from~$S$.

Since the variable-first procedure selects a variable with probability $\frac{1}{15}$, switching strategy from $S$ to $S'$ changes the \emph{overall} success probability by $-\frac{2}{45}$, $-\frac{1}{45}$, $0$, $\frac{1}{45}$, or $\frac{2}{45}$. However, as $S$ is already an optimal strategy, the change in probability cannot be $\frac{1}{45}$ or $\frac{2}{45}$. Moreover, the optimal doily game strategy is not symmetric, so the change in probability cannot be 0. Additionally, by \Autoref{thm:sym_prob}, $S'$ must win with probability of the form $\frac{2q + 15}{45}$. However, there is no $q$ such that $\frac{2q + 15}{45} = \frac{38}{45} = \frac{8}{9} - \frac{2}{45}$, so the change in probability also cannot be $-\frac{2}{45}$.

Thus, we need only consider a change in overall success probability by $-\frac{1}{45}$. As such, $S'$ wins with probability $\frac{8}{9} - \frac{1}{45} = \frac{13}{15}$, so it is an optimal symmetric strategy for the doily game. Furthermore, by \Autoref{thm:sym_prob}, $S'$ has three variables with $\frac{2}{3}$-consistent assignments.

Recall from our three cases for each variable that an overall change in success probability of $-\frac{1}{45}$ is only possible if one variable satisfies \autoref{case:2a} and the other variable satisfies \autoref{case:2b}. Without loss of generality, assume $V_m$ satisfies case 2a and $V_n$ satisfies case 2b. Now, consider the symmetric strategy $S'' = (A'', B'')$ in which $A'' = B''$ is nearly identical to $A'$, except that
\begin{align*}
    A''_m(j) &= \overline{A'_m(j)} & A''_n(j) &= \overline{A'_n(j)}.
\end{align*}
Since the only changes between $A'$ and $A''$ are two bits flipped in the response to $E_j$, $S''$ still satisfies all parity constraints. Furthermore, since $V_m$ satisfies case 2a, $V_m$ still has a $\frac{2}{3}$-consistent assignment under $S''$. On the other hand, since $V_n$ satisfies case 2b, $V_n$ has a fully consistent assignment under $S''$. As $S'$ has exactly three variables with $\frac{2}{3}$-consistent assignments, including $V_m$ and $V_n$, $S''$ has exactly two variables with $\frac{2}{3}$-consistent assignments. By \Autoref{thm:sym_prob}, $S''$ wins the doily game with probability $\frac{41}{45} > \frac{39}{45} = \frac{13}{15}$, which contradicts \Autoref{thm:opt_sym}.

Thus, the original assumption that $S$ wins on more than 13 out of the 15 synchronous questions is false, completing the proof of \Autoref{lem:sync}.
\end{proof}
Now, we may proceed to the proof of the main theorem:
\begin{proof}[Proof of \Autoref{thm:psams}]
Let $S_1$ be a symmetric strategy and $S_2$ an asymmetric for the $p$-sync doily game. As $S_1$ is a symmetric strategy, it always wins when the synchronous questions are asked, and, by \Autoref{thm:opt_sym}, it wins the non-synchronous, doily game questions with probability at most $\frac{13}{15}$. Thus, $S_1$ wins the $p$-sync doily game with probability at most
\begin{align}
    p + (1-p)\frac{13}{15}.
\end{align}

As $S_2$ is an asymmetric strategy, it wins the doily game questions with probability at most $\frac{8}{9}$, according to \Autoref{thm:ams_val}. Let $q \in [15]$ denote the number of equations on which Alice and Bob agree when using $S_2$. By \Autoref{lem:sync}, $q \leq 13$. As the probability that $S_2$ wins when the synchronous questions are asked is $\frac{q}{15}$, $S_2$ wins the $p$-sync doily game with probability at most
\begin{align}
    p\frac{13}{15} + (1-p)\frac{8}{9}.
\end{align}
We may then consider the $p$ such that the upper bound on the success probability of $S_1$ is the same as that of $S_2$, which is when
\begin{align}
    p + (1-p)\frac{13}{15} = p\frac{13}{15} + (1-p)\frac{8}{9}.
\end{align}
Solving the above yields $p = \frac{1}{7}$, and the upper bound on the success probability for the game using either kind of strategy is $\frac{31}{35} < \frac{8}{9}$. Moreover, a success probability of $\frac{31}{35}$ is exactly achieved by an optimal symmetric strategy for the doily game, which completes the proof of \Autoref{thm:psams}.
\end{proof}

\section{Conclusion}

In this work, we have shown that the $\frac{1}{7}$-sync doily game improves the entangled-classical gap for two Bell states to $\frac{4}{35}$ from the previous known value of $\frac{1}{9}$ for the MS game. We close this work with several open questions.

Though we have found the $\frac{1}{7}$-sync doily game to surpass the MS game in distinguishing between two Bell states and no entanglement, there may be yet another game that surpasses the $\frac{1}{7}$-sync doily game. In particular, from \Autoref{thm:psams}, we have that for $p = \frac{1}{7}$, the the classical value for the $p$-synchronous doily game is $\frac{31}{35} < \frac{8}{9}$. However, it is possible that for $p < \frac{1}{7}$, the value of the $p$-synchronous doily game is less than $\frac{31}{35}$. In particular, if \Autoref{lem:sync} can be improved to state that an optimal asymmetric strategy for the doily game wins on at most 12 of 15 synchronous questions, the $\frac{1}{10}$-synchronous doily game would have a classical value of $\frac{22}{25} < \frac{31}{35}$. The author estimates that this could be checked exhaustively using about three years of continuous work on an RTX 4090 graphics card.

Additionally, the doily game and the $p$-sync doily game can easily be generalized to higher dimensions in several ways. One way, for a chosen dimension $2^d$, is that we may associate each of the commuting triples of $P_d$ having product $\pm I^{\otimes d}$ with an equation and each non-identity element of $P_d$ with a variable. Such a game may be won perfectly by players using maximal entanglement with local dimension $2^d$, but how well can unentangled players do? In particular, can such a game for $d = 4$ better distinguish entangled and unentangled players than parallel repetition of the MS game?

Like the MS game, the doily game (and $p$-sync doily game) admits an equivalence class of games obtained by multiplying operators by $-II$ and accordingly flipping the associated parity requirements. For the MS game, equivalence can be checked according to the property that the total parity of the magic square according to the rows must differ from the total parity of the magic square according to the columns for the game to be an MS-equivalent game. However, the equations of the doily and game cannot be partitioned into rows and columns. Is there a simple property of the parity constraints that can be checked to determine whether a game is equivalent to the doily game?

\section{Acknowledgments}

I would like to thank Richard Cleve for numerous insightful discussions that helped to develop this work and to make this manuscript presentable. I would also like to thank Michael Tiemann for pointing out reference \cite{marceaux2019}, which helped to visualize the three dimensional geometry of the doily and $p$-sync doily games. This research received no external funding.

{
\sloppy
\printbibliography
}

\pagebreak
\appendix

\section{Operator magic squares} \label{appx:squares}
In \autoref{fig:oms}, we list the nine operator magic squares that form the complete set of 10 with the original given in \autoref{fig:ms_diagr}.

\begin{figure}[h]
    \centering
    \begin{tikzpicture}[scale=0.4]
    \coordinate (V0) at (-2.75,2.75);
    \coordinate (V1) at (0,2.75);
    \coordinate (V2) at (2.75,2.75);
    \coordinate (V3) at (-2.75,0);
    \coordinate (V4) at (0,0);
    \coordinate (V5) at (2.75,0);
    \coordinate (V6) at (-2.75,-2.75);
    \coordinate (V7) at (0,-2.75);
    \coordinate (V8) at (2.75,-2.75);
    
    \draw[ultra thick] (V0) -- (V2);
    \draw[ultra thick] (V3) -- (V5);
    \draw[ultra thick] (V6) -- (V8);
    
    \draw[red,ultra thick] (V0) -- (V6);
    \draw[ultra thick] (V1) -- (V7);
    \draw[ultra thick] (V2) -- (V8);
    
    \draw[ultra thick,fill=white] (V0) circle [radius=1];
    \draw[ultra thick,fill=white] (V1) circle [radius=1];
    \draw[ultra thick,fill=white] (V2) circle [radius=1];
    \draw[ultra thick,fill=white] (V3) circle [radius=1];
    \draw[ultra thick,fill=white] (V4) circle [radius=1];
    \draw[ultra thick,fill=white] (V5) circle [radius=1];
    \draw[ultra thick,fill=white] (V6) circle [radius=1];
    \draw[ultra thick,fill=white] (V7) circle [radius=1];
    \draw[ultra thick,fill=white] (V8) circle [radius=1];
    
    \node at (V0) {$XX$};
    \node at (V1) {$YZ$};
    \node at (V2) {$ZY$};
    \node at (V3) {$YY$};
    \node at (V4) {$YI$};
    \node at (V5) {$IY$};
    \node at (V6) {$ZZ$};
    \node at (V7) {$IZ$};
    \node at (V8) {$ZI$};
    
    \node at (-4.8,2.75) {\Large $E_9$};
    \node at (-4.8,0) {\Large $E_4$};
    \node at (-4.8,-2.75) {\Large $E_8$};
    
    \node at (-2.75,4.8) {\Large $E_{12}$};
    \node at (0,4.8) {\Large $E_5$};
    \node at (2.75,4.8) {\Large $E_7$};
    \end{tikzpicture}
    \begin{tikzpicture}[scale=0.4]
    \coordinate (V0) at (-2.75,2.75);
    \coordinate (V1) at (0,2.75);
    \coordinate (V2) at (2.75,2.75);
    \coordinate (V3) at (-2.75,0);
    \coordinate (V4) at (0,0);
    \coordinate (V5) at (2.75,0);
    \coordinate (V6) at (-2.75,-2.75);
    \coordinate (V7) at (0,-2.75);
    \coordinate (V8) at (2.75,-2.75);
    
    \draw[ultra thick] (V0) -- (V2);
    \draw[ultra thick] (V3) -- (V5);
    \draw[ultra thick] (V6) -- (V8);
    
    \draw[red,ultra thick] (V0) -- (V6);
    \draw[ultra thick] (V1) -- (V7);
    \draw[ultra thick] (V2) -- (V8);
    
    \draw[ultra thick,fill=white] (V0) circle [radius=1];
    \draw[ultra thick,fill=white] (V1) circle [radius=1];
    \draw[ultra thick,fill=white] (V2) circle [radius=1];
    \draw[ultra thick,fill=white] (V3) circle [radius=1];
    \draw[ultra thick,fill=white] (V4) circle [radius=1];
    \draw[ultra thick,fill=white] (V5) circle [radius=1];
    \draw[ultra thick,fill=white] (V6) circle [radius=1];
    \draw[ultra thick,fill=white] (V7) circle [radius=1];
    \draw[ultra thick,fill=white] (V8) circle [radius=1];
    
    \node at (V0) {$XX$};
    \node at (V1) {$IX$};
    \node at (V2) {$XI$};
    \node at (V3) {$YY$};
    \node at (V4) {$ZX$};
    \node at (V5) {$XZ$};
    \node at (V6) {$ZZ$};
    \node at (V7) {$ZI$};
    \node at (V8) {$IZ$};
    
    \node at (-4.8,2.75) {\Large $E_0$};
    \node at (-4.8,0) {\Large $E_{10}$};
    \node at (-4.8,-2.75) {\Large $E_8$};
    
    \node at (-2.75,4.8) {\Large $E_{12}$};
    \node at (0,4.8) {\Large $E_6$};
    \node at (2.75,4.8) {\Large $E_2$};
    \end{tikzpicture}
    \begin{tikzpicture}[scale=0.4]
    \coordinate (V0) at (-2.75,2.75);
    \coordinate (V1) at (0,2.75);
    \coordinate (V2) at (2.75,2.75);
    \coordinate (V3) at (-2.75,0);
    \coordinate (V4) at (0,0);
    \coordinate (V5) at (2.75,0);
    \coordinate (V6) at (-2.75,-2.75);
    \coordinate (V7) at (0,-2.75);
    \coordinate (V8) at (2.75,-2.75);
    
    \draw[ultra thick] (V0) -- (V2);
    \draw[ultra thick] (V3) -- (V5);
    \draw[ultra thick] (V6) -- (V8);
    
    \draw[red,ultra thick] (V0) -- (V6);
    \draw[ultra thick] (V1) -- (V7);
    \draw[ultra thick] (V2) -- (V8);
    
    \draw[ultra thick,fill=white] (V0) circle [radius=1];
    \draw[ultra thick,fill=white] (V1) circle [radius=1];
    \draw[ultra thick,fill=white] (V2) circle [radius=1];
    \draw[ultra thick,fill=white] (V3) circle [radius=1];
    \draw[ultra thick,fill=white] (V4) circle [radius=1];
    \draw[ultra thick,fill=white] (V5) circle [radius=1];
    \draw[ultra thick,fill=white] (V6) circle [radius=1];
    \draw[ultra thick,fill=white] (V7) circle [radius=1];
    \draw[ultra thick,fill=white] (V8) circle [radius=1];
    
    \node at (V0) {$XX$};
    \node at (V1) {$XI$};
    \node at (V2) {$IX$};
    \node at (V3) {$YY$};
    \node at (V4) {$IY$};
    \node at (V5) {$YI$};
    \node at (V6) {$ZZ$};
    \node at (V7) {$XY$};
    \node at (V8) {$YX$};
    
    \node at (-4.8,2.75) {\Large $E_0$};
    \node at (-4.8,0) {\Large $E_4$};
    \node at (-4.8,-2.75) {\Large $E_{11}$};
    
    \node at (-2.75,4.8) {\Large $E_{12}$};
    \node at (0,4.8) {\Large $E_1$};
    \node at (2.75,4.8) {\Large $E_3$};
    \end{tikzpicture} \\
    \vspace{0.15cm}
    \begin{tikzpicture}[scale=0.4]
    \coordinate (V0) at (-2.75,2.75);
    \coordinate (V1) at (0,2.75);
    \coordinate (V2) at (2.75,2.75);
    \coordinate (V3) at (-2.75,0);
    \coordinate (V4) at (0,0);
    \coordinate (V5) at (2.75,0);
    \coordinate (V6) at (-2.75,-2.75);
    \coordinate (V7) at (0,-2.75);
    \coordinate (V8) at (2.75,-2.75);
    
    \draw[ultra thick] (V0) -- (V2);
    \draw[ultra thick] (V3) -- (V5);
    \draw[ultra thick] (V6) -- (V8);
    
    \draw[ultra thick] (V0) -- (V6);
    \draw[red,ultra thick] (V1) -- (V7);
    \draw[ultra thick] (V2) -- (V8);
    
    \draw[ultra thick,fill=white] (V0) circle [radius=1];
    \draw[ultra thick,fill=white] (V1) circle [radius=1];
    \draw[ultra thick,fill=white] (V2) circle [radius=1];
    \draw[ultra thick,fill=white] (V3) circle [radius=1];
    \draw[ultra thick,fill=white] (V4) circle [radius=1];
    \draw[ultra thick,fill=white] (V5) circle [radius=1];
    \draw[ultra thick,fill=white] (V6) circle [radius=1];
    \draw[ultra thick,fill=white] (V7) circle [radius=1];
    \draw[ultra thick,fill=white] (V8) circle [radius=1];
    
    \node at (V0) {$XX$};
    \node at (V1) {$YZ$};
    \node at (V2) {$ZY$};
    \node at (V3) {$IX$};
    \node at (V4) {$ZX$};
    \node at (V5) {$ZI$};
    \node at (V6) {$XI$};
    \node at (V7) {$XY$};
    \node at (V8) {$IY$};
    
    \node at (-4.8,2.75) {\Large $E_9$};
    \node at (-4.8,0) {\Large $E_6$};
    \node at (-4.8,-2.75) {\Large $E_1$};
    
    \node at (-2.75,4.8) {\Large $E_0$};
    \node at (0,4.8) {\Large $E_{13}$};
    \node at (2.75,4.8) {\Large $E_7$};
    \end{tikzpicture}
    \begin{tikzpicture}[scale=0.4]
    \coordinate (V0) at (-2.75,2.75);
    \coordinate (V1) at (0,2.75);
    \coordinate (V2) at (2.75,2.75);
    \coordinate (V3) at (-2.75,0);
    \coordinate (V4) at (0,0);
    \coordinate (V5) at (2.75,0);
    \coordinate (V6) at (-2.75,-2.75);
    \coordinate (V7) at (0,-2.75);
    \coordinate (V8) at (2.75,-2.75);
    
    \draw[ultra thick] (V0) -- (V2);
    \draw[ultra thick] (V3) -- (V5);
    \draw[ultra thick] (V6) -- (V8);
    
    \draw[ultra thick] (V0) -- (V6);
    \draw[red,ultra thick] (V1) -- (V7);
    \draw[ultra thick] (V2) -- (V8);
    
    \draw[ultra thick,fill=white] (V0) circle [radius=1];
    \draw[ultra thick,fill=white] (V1) circle [radius=1];
    \draw[ultra thick,fill=white] (V2) circle [radius=1];
    \draw[ultra thick,fill=white] (V3) circle [radius=1];
    \draw[ultra thick,fill=white] (V4) circle [radius=1];
    \draw[ultra thick,fill=white] (V5) circle [radius=1];
    \draw[ultra thick,fill=white] (V6) circle [radius=1];
    \draw[ultra thick,fill=white] (V7) circle [radius=1];
    \draw[ultra thick,fill=white] (V8) circle [radius=1];
    
    \node at (V0) {$YI$};
    \node at (V1) {$YZ$};
    \node at (V2) {$IZ$};
    \node at (V3) {$YY$};
    \node at (V4) {$ZX$};
    \node at (V5) {$XZ$};
    \node at (V6) {$IY$};
    \node at (V7) {$XY$};
    \node at (V8) {$XI$};
    
    \node at (-4.8,2.75) {\Large $E_5$};
    \node at (-4.8,0) {\Large $E_{10}$};
    \node at (-4.8,-2.75) {\Large $E_1$};
    
    \node at (-2.75,4.8) {\Large $E_4$};
    \node at (0,4.8) {\Large $E_{13}$};
    \node at (2.75,4.8) {\Large $E_2$};
    \end{tikzpicture}
    \begin{tikzpicture}[scale=0.4]
    \coordinate (V0) at (-2.75,2.75);
    \coordinate (V1) at (0,2.75);
    \coordinate (V2) at (2.75,2.75);
    \coordinate (V3) at (-2.75,0);
    \coordinate (V4) at (0,0);
    \coordinate (V5) at (2.75,0);
    \coordinate (V6) at (-2.75,-2.75);
    \coordinate (V7) at (0,-2.75);
    \coordinate (V8) at (2.75,-2.75);
    
    \draw[ultra thick] (V0) -- (V2);
    \draw[ultra thick] (V3) -- (V5);
    \draw[ultra thick] (V6) -- (V8);
    
    \draw[ultra thick] (V0) -- (V6);
    \draw[red,ultra thick] (V1) -- (V7);
    \draw[ultra thick] (V2) -- (V8);
    
    \draw[ultra thick,fill=white] (V0) circle [radius=1];
    \draw[ultra thick,fill=white] (V1) circle [radius=1];
    \draw[ultra thick,fill=white] (V2) circle [radius=1];
    \draw[ultra thick,fill=white] (V3) circle [radius=1];
    \draw[ultra thick,fill=white] (V4) circle [radius=1];
    \draw[ultra thick,fill=white] (V5) circle [radius=1];
    \draw[ultra thick,fill=white] (V6) circle [radius=1];
    \draw[ultra thick,fill=white] (V7) circle [radius=1];
    \draw[ultra thick,fill=white] (V8) circle [radius=1];
    
    \node at (V0) {$IZ$};
    \node at (V1) {$YZ$};
    \node at (V2) {$YI$};
    \node at (V3) {$ZI$};
    \node at (V4) {$ZX$};
    \node at (V5) {$IX$};
    \node at (V6) {$ZZ$};
    \node at (V7) {$XY$};
    \node at (V8) {$YX$};
    
    \node at (-4.8,2.75) {\Large $E_5$};
    \node at (-4.8,0) {\Large $E_6$};
    \node at (-4.8,-2.75) {\Large $E_{11}$};
    
    \node at (-2.75,4.8) {\Large $E_8$};
    \node at (0,4.8) {\Large $E_{13}$};
    \node at (2.75,4.8) {\Large $E_3$};
    \end{tikzpicture} \\
    \vspace{0.15cm}
    \begin{tikzpicture}[scale=0.4]
    \coordinate (V0) at (-2.75,2.75);
    \coordinate (V1) at (0,2.75);
    \coordinate (V2) at (2.75,2.75);
    \coordinate (V3) at (-2.75,0);
    \coordinate (V4) at (0,0);
    \coordinate (V5) at (2.75,0);
    \coordinate (V6) at (-2.75,-2.75);
    \coordinate (V7) at (0,-2.75);
    \coordinate (V8) at (2.75,-2.75);
    
    \draw[ultra thick] (V0) -- (V2);
    \draw[ultra thick] (V3) -- (V5);
    \draw[ultra thick] (V6) -- (V8);
    
    \draw[ultra thick] (V0) -- (V6);
    \draw[ultra thick] (V1) -- (V7);
    \draw[red,ultra thick] (V2) -- (V8);
    
    \draw[ultra thick,fill=white] (V0) circle [radius=1];
    \draw[ultra thick,fill=white] (V1) circle [radius=1];
    \draw[ultra thick,fill=white] (V2) circle [radius=1];
    \draw[ultra thick,fill=white] (V3) circle [radius=1];
    \draw[ultra thick,fill=white] (V4) circle [radius=1];
    \draw[ultra thick,fill=white] (V5) circle [radius=1];
    \draw[ultra thick,fill=white] (V6) circle [radius=1];
    \draw[ultra thick,fill=white] (V7) circle [radius=1];
    \draw[ultra thick,fill=white] (V8) circle [radius=1];
    
    \node at (V0) {$XX$};
    \node at (V1) {$YZ$};
    \node at (V2) {$ZY$};
    \node at (V3) {$XI$};
    \node at (V4) {$IZ$};
    \node at (V5) {$XZ$};
    \node at (V6) {$IX$};
    \node at (V7) {$YI$};
    \node at (V8) {$YX$};
    
    \node at (-4.8,2.75) {\Large $E_9$};
    \node at (-4.8,0) {\Large $E_2$};
    \node at (-4.8,-2.75) {\Large $E_3$};
    
    \node at (-2.75,4.8) {\Large $E_0$};
    \node at (0,4.8) {\Large $E_5$};
    \node at (2.75,4.8) {\Large $E_{14}$};
    \end{tikzpicture}
    \begin{tikzpicture}[scale=0.4]
    \coordinate (V0) at (-2.75,2.75);
    \coordinate (V1) at (0,2.75);
    \coordinate (V2) at (2.75,2.75);
    \coordinate (V3) at (-2.75,0);
    \coordinate (V4) at (0,0);
    \coordinate (V5) at (2.75,0);
    \coordinate (V6) at (-2.75,-2.75);
    \coordinate (V7) at (0,-2.75);
    \coordinate (V8) at (2.75,-2.75);
    
    \draw[ultra thick] (V0) -- (V2);
    \draw[ultra thick] (V3) -- (V5);
    \draw[ultra thick] (V6) -- (V8);
    
    \draw[ultra thick] (V0) -- (V6);
    \draw[ultra thick] (V1) -- (V7);
    \draw[red,ultra thick] (V2) -- (V8);
    
    \draw[ultra thick,fill=white] (V0) circle [radius=1];
    \draw[ultra thick,fill=white] (V1) circle [radius=1];
    \draw[ultra thick,fill=white] (V2) circle [radius=1];
    \draw[ultra thick,fill=white] (V3) circle [radius=1];
    \draw[ultra thick,fill=white] (V4) circle [radius=1];
    \draw[ultra thick,fill=white] (V5) circle [radius=1];
    \draw[ultra thick,fill=white] (V6) circle [radius=1];
    \draw[ultra thick,fill=white] (V7) circle [radius=1];
    \draw[ultra thick,fill=white] (V8) circle [radius=1];
    
    \node at (V0) {$IY$};
    \node at (V1) {$ZI$};
    \node at (V2) {$ZY$};
    \node at (V3) {$YY$};
    \node at (V4) {$ZX$};
    \node at (V5) {$XZ$};
    \node at (V6) {$YI$};
    \node at (V7) {$IX$};
    \node at (V8) {$YX$};
    
    \node at (-4.8,2.75) {\Large $E_7$};
    \node at (-4.8,0) {\Large $E_{10}$};
    \node at (-4.8,-2.75) {\Large $E_3$};
    
    \node at (-2.75,4.8) {\Large $E_4$};
    \node at (0,4.8) {\Large $E_6$};
    \node at (2.75,4.8) {\Large $E_{14}$};
    \end{tikzpicture}
    \begin{tikzpicture}[scale=0.4]
    \coordinate (V0) at (-2.75,2.75);
    \coordinate (V1) at (0,2.75);
    \coordinate (V2) at (2.75,2.75);
    \coordinate (V3) at (-2.75,0);
    \coordinate (V4) at (0,0);
    \coordinate (V5) at (2.75,0);
    \coordinate (V6) at (-2.75,-2.75);
    \coordinate (V7) at (0,-2.75);
    \coordinate (V8) at (2.75,-2.75);
    
    \draw[ultra thick] (V0) -- (V2);
    \draw[ultra thick] (V3) -- (V5);
    \draw[ultra thick] (V6) -- (V8);
    
    \draw[ultra thick] (V0) -- (V6);
    \draw[ultra thick] (V1) -- (V7);
    \draw[red,ultra thick] (V2) -- (V8);
    
    \draw[ultra thick,fill=white] (V0) circle [radius=1];
    \draw[ultra thick,fill=white] (V1) circle [radius=1];
    \draw[ultra thick,fill=white] (V2) circle [radius=1];
    \draw[ultra thick,fill=white] (V3) circle [radius=1];
    \draw[ultra thick,fill=white] (V4) circle [radius=1];
    \draw[ultra thick,fill=white] (V5) circle [radius=1];
    \draw[ultra thick,fill=white] (V6) circle [radius=1];
    \draw[ultra thick,fill=white] (V7) circle [radius=1];
    \draw[ultra thick,fill=white] (V8) circle [radius=1];
    
    \node at (V0) {$ZI$};
    \node at (V1) {$IY$};
    \node at (V2) {$ZY$};
    \node at (V3) {$IZ$};
    \node at (V4) {$XI$};
    \node at (V5) {$XZ$};
    \node at (V6) {$ZZ$};
    \node at (V7) {$XY$};
    \node at (V8) {$YX$};
    
    \node at (-4.8,2.75) {\Large $E_7$};
    \node at (-4.8,0) {\Large $E_2$};
    \node at (-4.8,-2.75) {\Large $E_{11}$};
    
    \node at (-2.75,4.8) {\Large $E_8$};
    \node at (0,4.8) {\Large $E_1$};
    \node at (2.75,4.8) {\Large $E_{14}$};
    \end{tikzpicture}
    \caption{The nine remaining operator magic squares.}
    \label{fig:oms}
\end{figure}

\pagebreak
\section{Satisfiability program} \label{appx:sat}

Below is the Python script used with Python 3.11.1 to determine the maximum number of doily game equations satisfiable by a fixed bit assignment. The script is used to prove \Autoref{thm:eqn_sat}.

\begin{minted}{Python}
var_names = [
    "ix", "iy", "iz", "xi", "xx",
    "xy", "xz", "yi", "yx", "yy",
    "yz", "zi", "zx", "zy", "zz",
]
equations = [
    ("ix", "xi", "xx", 0),
    ("iy", "xi", "xy", 0),
    ("iz", "xi", "xz", 0),
    ("ix", "yi", "yx", 0),
    ("iy", "yi", "yy", 0),
    ("iz", "yi", "yz", 0),
    ("ix", "zi", "zx", 0),
    ("iy", "zi", "zy", 0),
    ("iz", "zi", "zz", 0),
    ("yz", "zy", "xx", 0),
    ("zx", "xz", "yy", 0),
    ("xy", "yx", "zz", 0),
    ("xx", "yy", "zz", 1),
    ("xy", "yz", "zx", 1),
    ("xz", "yx", "zy", 1),
]

max_sat = 0
for j in range(1 << len(var_names)):
    # get bit assignment from integer
    vals = {}
    for k, name in enumerate(var_names):
        vals[name] = (j >> k) % 2
    # count satisfied equations
    sat = 0
    for var1, var2, var3, result in equations:
        if (vals[var1] + vals[var2] + vals[var3]) % 2 == result:
            sat += 1
    # check if new max
    if sat > max_sat:
        max_sat = sat

print("Maximum number of satisfied equations: ", max_sat)
\end{minted}
\end{document}